\documentclass[11pt,cls,onecolumn]{IEEEtran}

\usepackage{stfloats}
\usepackage{color}
\usepackage{amssymb} 
\usepackage{mathtools}
\usepackage{footnote}
\usepackage{amsfonts}

\newcommand{\ve}[1]{{\bf #1}}

\newcommand{\pp}[1]{\ensuremath{\mathbb P\left\{#1\right\}}}

\newtheorem{lemma}{Lemma}
\newtheorem{theorem}{Theorem}

\newtheorem{definition}{Definition}
\newtheorem{corollary}{Corollary}
\newtheorem{remark}{Remark}

\newcommand{\E}[1]{{\rm I\kern-.3em E}\left[#1\right]}
\newcommand{\var}[1]{\mbox{Var}\left[#1\right]}

\ifCLASSOPTIONcompsoc
\usepackage[tight,normalsize,sf,SF]{subfigure}
\else
\usepackage[tight,footnotesize]{subfigure}
\fi

\begin{document}

\sloppy
\IEEEoverridecommandlockouts
\title{Typical sumsets of linear codes}

\author{Jingge Zhu and Michael Gastpar\IEEEmembership{, Member, IEEE}
\thanks{This work was supported in part by the European ERC Starting Grant 259530-ComCom. }
\thanks{J. Zhu and M. Gastpar are with the School of Computer and Communication Sciences, Ecole Polytechnique F{\'e}d{\'e}rale de Lausanne (EPFL), Lausanne,
Switzerland (e-mail: jingge.zhu@epfl.ch, michael.gastpar@epfl.ch).}
%
%
}




\maketitle

\begin{abstract}
Given two identical linear codes $\mathcal C$ over $\mathbb F_q$ of length $n$,  we independently pick one codeword from each codebook uniformly at random. A \textit{sumset} is formed by adding these two codewords entry-wise as integer vectors and a sumset is called \textit{typical}, if the sum  falls inside this set with high probability.  We ask the question: how large is the typical sumset for most codes? In this paper we characterize the asymptotic size of such typical sumset. We show that when the rate $R$ of the linear code is below a certain threshold $D$, the typical sumset size is roughly $|\mathcal C|^2=2^{2nR}$ for most codes while when $R$ is above this threshold, most codes have a typical sumset whose size is roughly $|\mathcal C|\cdot 2^{nD}=2^{n(R+D)}$ due to the linear structure of the codes.  The threshold $D$ depends solely on the alphabet size $q$ and takes value in $[1/2, \log \sqrt{e})$.  More generally, we completely characterize the  asymptotic size of  typical sumsets of two nested linear codes $\mathcal C_1, \mathcal C_2$ with different rates.   As an application of the result, we study the communication problem where the integer sum of two codewords is to be decoded through a general two-user multiple-access channel.  
\end{abstract}

\section{Introduction}
\label{sec:intro}

Structured codes (linear codes for example) not only permits  simple encoding and decoding algorithms, but also provides good interference mitigation properties which are crucial for multi-user communication networks.  Specialized to Gaussian wireless networks, lattice codes, which can be seen as linear codes (which are the most well-understood structured codes) lifted to Euclidean space \cite{Erez_etal_2005}, have been studied extensively.  Early results on lattice codes including \cite{loeliger_averaging_1997} \cite{urbanke_lattice_1998} \cite{ErezZamir_2004} have shown that good (nested) lattice codes are able to achieve the capacity of point-to-point Gaussian channels. Lattice codes are also applied to Gaussian networks, for example the Gaussian two-way relay channel (\cite{Nam_etal_2010}\cite{wilson_joint_2010}), and yield best known communication rates that cannot be achieved otherwise. More recently, the \textit{compute-and-forward} \cite{NazerGastpar_2011} framework employs nested lattice codes in a general Gaussian wireless network. It exploits the additivity of the network by  addressing the problem of decoding  sums of lattice codewords at  intermediate nodes in the network. Furthermore  nested linear codes (see  \cite{padakandla_computing_2013}, \cite{miyake_2010} for example), which can be seen as a generalization of nested lattice codes, are applicable to  general multi-user networks other than Gaussian networks.

Consider applying the simplest structured codes --- linear codes,  to a standard two-user Gaussian multiple access channel (MAC) of the form $Y=X_1+X_2+Z$. Existing coding schemes using structured codes usually consider two codewords $T_1^n, T_2^n$ in some vector space over a finite field, say $\mathbb F_q^n$, and require the entry-wise modulo sum $T_1^n\oplus T_2^n$ to be decoded at the receiver.  But for the Gaussian MAC it is more natural to study the ``integer sum" $T_1^n+T_2^n$, where two codewords are treated as integer-valued vectors. This is because after lifting  linear codes from the $\mathbb F_q^n$ to $\mathbb R^n$, the additive Gaussian channel sums up $T_1^n, T_2^n$ as vectors of real numbers instead of in a finite field.  The modulo sum $T_1^n\oplus T_2^n$ is easy to understand: if $T_1^n, T_2^n$ are uniformly chosen from a linear code, the sum $T_1^n\oplus T_2^n$ stays in that linear code and is still uniformly distributed. But the analysis of the integer sum $T_1^n+T_2^n$ is more complicated and its behavior have not been studied. 

To put our study in perspective, it is worth pointing out that our problem is closely connected to \textit{sumset theory}, which studies the size of the set $\mathcal A+\mathcal B:=\{a+b: a\in\mathcal  A, b\in\mathcal  B\}$ where  $\mathcal A,\mathcal B$ are two finite sets taking values in some additive group. One objective of the sumset theory is to use \textit{sumset inequalities} to relate the cardinality of sets $|\mathcal A|, |\mathcal B|$ and $|\mathcal A+\mathcal B|$. As a simple example, for $\mathcal A=\{0,1,2,3,4\}$ with $5$ elements we have $|\mathcal A+\mathcal A|=9$ elements. But if  let $\mathcal A'=\{0,0.2,0.8,1.1,2.1\}$ with $5$ elements we have $|\mathcal A'+\mathcal A'|=15$ elements. This shows that the sumset size $|\mathcal A+\mathcal B|$ depends heavily on structures of the sets. As a rule of thumb, the sumset size will be small if and only if the individual sets are ``structured".  Some classical results of sumset theory and inverse sumset theory can be found in, e.g.  \cite{ruzsa_sumsets_2009}.

Our problem concerns with  \textit{sums of random variables} defined over a certain set,  hence can be viewed as a sumset problem in a probabilistic setting. It shares similarity with the classical sumset problem while has its  own feature. We first point out the main difference between the two problems. Given a set of integers $\mathcal U=\{0,1,\ldots,q-1\}$,  the sumset $\mathcal U+\mathcal U$ contains $2q-1$ elements. Now let  $U_1, U_2$ be two independent random variables uniformly distributed in the set $\mathcal U$, a natural connection between the size of the set $\mathcal U$ and the random variables $U_1, U_2$ is that $H(U_1)=H(U_2)=\log |\mathcal U|$, i.e., the entropy of the random variable is equal to the logarithmic size of $\mathcal U$. Now we turn to the sum variable $W:=U_1+U_2$. Although $W$ takes \textit{all} possible values in $\mathcal U+\mathcal U$,  it is   ``smaller" than $\log|\mathcal U+\mathcal U|$  because the distribution of $W$ is non-uniform over $\mathcal U+\mathcal U$. Indeed we have $H(W)< \log|\mathcal U+\mathcal U|$ in this case but the difference between $H(W)$ and $\log |\mathcal U+\mathcal U|$ is small.  However this phenomenon is much more pronounced in high dimensional spaces as we shall see later in this paper.  On the other hand it is also important to realize that in the probabilistic setting, the structure of the random variable still has decisive impact on the sumset ``size", which can be partially characterized by the entropy of the sum variable. Using the examples in the preceding paragraph, if the identical independent random variables $U_1,U_2$ are  uniformly distributed in $\mathcal A$, we have $H(U_1+U_2)\approx 2.99$ bit while if $U'_1,U'_2$ uniformly distributed in $\mathcal A'$, it gives $H(U'_1+U'_2)\approx 3.84$ bit.  We also point out that the sumset theory for  Shannon entropy has been studied recently in e.g. \cite{tao_sumset_2010}  \cite{kontoyiannis_sumset_2014} and fundamental results relating $H(X)$ and $H(X_1+X_2)$ are established. However our specific problem about linear codes in high-dimensional spaces requires separate analysis which is not present in the existing literature.




In this paper, we consider two linear codes $\mathcal C_1, \mathcal C_2$ with rates $R_1, R_2$ while satisfying the condition $\mathcal C_1\subseteq\mathcal C_2$ or $\mathcal C_2\subseteq\mathcal C_1$. Let $T_1^n,T_2^n$ be   two codewords uniformly chosen from $\mathcal C_1,\mathcal C_2$ and we would like to understand what does the sum $W^n:=T_1^n+T_2^n$ look like in $\mathbb Z^n$ for very large $n$. We will show that when the dimension $n$ goes to infinity, most sums $T_1^n+T_2^n$ will fall into a subset $\mathcal K$, which could be substantially smaller than the sumset $\mathcal C_1+\mathcal C_2$.  We characterize the asymptotic size of $\mathcal K$ completely and show certain thresholds effects of the size $|\mathcal K|$ depending on the values of $R_1, R_2$. We also established the exact relationship between the $H(T_1^n,H_2^n)$ and $H(T_1^n+T_2^n)$ in the limit and show that the difference between $H(T_1^n+T_2^n)$ and $\log|\mathcal C_1+\mathcal C_2 |$ can increase unboundedly as the codewod length $n$ increases.  As an application of the results, we study the problem of decoding the integer sum of codewords through a general two-user Gaussian MAC when two users are equipped with two linear codes.

\section{Typical sumsets of linear codes}
\label{sec:sumsets}
In this section we  formally define and study typical sumsets of linear codes. 

\subsection{Preliminaries and notations}
\label{sec:notation}
We use $[a:b]$ to denote the set of integers $\{a,a+1,\ldots, b-1, b\}$ and define two sets $\mathcal U:=[0:q-1]$ and $\mathcal W:=[0:2q-2]$. We also define $P_U$ to be the uniform probability distribution over the  set $\mathcal U$  i.e., 
\begin{align}
P_U(a)=1/q \mbox{ for all } a\in\mathcal U.
\label{eq:P_U}
\end{align}
If $U_1, U_2$ are two independent random variables with distribution $P_U$,  the sum $W:=U_1+U_2$ is a random variable distributed over the  set $\mathcal W$.  Let $P_W$ denote the probability distribution of this random variable. A direct calculation shows that 
\begin{align}
P_W(a)=\begin{cases}
\frac{a+1}{q^2} & a\in[0:q-1]\\
\frac{2q-1-a}{q^2} &a\in[q:2q-2]
\end{cases}
\label{eq:P_W}
\end{align}
and the entropy of $W$ is given as
\begin{align}
H(W)=2\log q-\frac{1}{q^2}(2\sum_{i=1}^qi\log i-q\log q).
\label{eq:H_W}
\end{align}
Given a probability distribution $P_U$ over the alphabet $\mathcal U$, we use $\mathcal A_{[U]}^{(n)}$ to denote the set of typical sequences defined as:
\begin{align}
\mathcal A_{[U]}^{(n)}:=\left\{\ve m:\left|P_U(a)-\frac{1}{n}N(a|\ve m)\right |\leq\delta, \mbox{ for all } a\in\mathcal U\right\}
\label{eq:typical_def}
\end{align}
where $N(a|\ve m)$ is the occurrence count of the symbol $a$ in sequence $\ve m=(\ve m_1,\ldots,\ve m_n)$. In the paper we will always choose $\delta$ small but satisfying $n\delta^2 \rightarrow \infty$  as $n\rightarrow \infty$. Similarly we can define the conditional typical sequences $\mathcal A_{[Z|U]}^{(n)}(\ve u)$ as well as the typical sequences $\mathcal A_{[Z U]}^{(n)}$ determined by a joint distribution $P_{ZU}$ as in \cite[Ch. 2]{csiszar_information_2011}. We recall the standard results regarding the typical sequences. 
\begin{lemma}[Typical sequences \cite{csiszar_information_2011}]
Let $U^n$ be a  $n$-length random vector with each entry i.i.d. according to $P_U$. Then for every $\delta>0$ in (\ref{eq:typical_def}), it holds that
\begin{align}
\pp{U^n\in\mathcal A_{[U]}^{(n)}}\geq 1-2|\mathcal U|e^{-2n\delta^2}
\end{align}
Furthermore, the size of set of typical sequences is bounded as
\begin{align}
2^{n(H(U)-\epsilon_n)}\leq |\mathcal A_{[U]}^{(n)}|\leq 2^{n(H(U)+\epsilon_n)}
\end{align}
for some $\epsilon_n\searrow 0$ as $n\rightarrow \infty$.
\label{lemma:typical_sequence}
\end{lemma}

In this paper vectors and matrices are denoted using bold letters such as $\ve a$ and $\ve A$, respectively. The $i$-th entry of a vector $\ve a$ is denoted as  $\ve a_i$ and $\ve A_i$ denotes the $i$-th column of the matrix $\ve A$. Throughout the paper, the notations $\ve A\ve b$ or $\ve a^T\ve b$ are understood as matrix multiplication \textit{modulo $q$}, or the matrix multiplication over the corresponding finite field. Modulo addition is denoted with $\oplus$ and $+$ means the usual addition over integers. Logarithm $\log$ is with base $2$. Sets are usually denoted using calligraphic letters such as $\mathcal A$ and their cardinality are denoted by $|\mathcal A|$. We often deal with quantities depending on the codeword length $n$. The notation $o_n(1)$ denotes a quantity that approaches $0$ as $n\rightarrow \infty$. We say $a\doteq 2^{nb}$ for some constant $b$ if there exists some $\epsilon_n\searrow 0$ such that $2^{n(b-\epsilon_n)}\leq a\leq 2^{n(b+\epsilon_n)}$. We also consider the probability of events in the limit when the codeword length $n$ goes to infinity. For any event $H$, we say the event $H$ occurs \textit{asymptotically almost surely} (a.a.s.) if $\pp{H}\rightarrow 1$ as $n\rightarrow\infty$. 

\subsection{Problem statement and main results}

Given two positive integers $k, n$ satisfying $k<n$, an $(n,k)$-linear code over $\mathbb F_q$ is a $k$-dimensional subspace in $\mathbb F_q^n$ where $q$ is a prime number. The \textit{rate} of this code is given by $R:=\frac{k}{n}\log q$. Any $(n,k)$-linear code can be constructed as
\begin{align}
\mathcal C=\left\{\ve t :\ve t=\ve G\ve m, \mbox{ for all }\ve m\in \mathbb F_q^k\right\}
\label{eq:linear_codes}
\end{align}
with a \textit{generator matrix} $\ve G\in\mathbb F_q^{n\times k}$ and $\ve m$ can be thought as a \textit{message}. An $(n,k)$-linear code $\mathcal C$ over $\mathbb F_q$ is called \textit{systematic} if it can be constructed as 
\begin{align}
\mathcal C=\left\{\ve t :\ve t=\begin{bmatrix}
\ve I_{k\times k}\\
\ve Q
\end{bmatrix}\ve m, \mbox{ for all }\ve m\in \mathbb F_q^k\right\}
\label{eq:systematic}
\end{align}
with some $\ve Q\in\mathbb F_q^{(n-k)\times k}$ where $\ve I_{k\times k}$ is the $k\times k$ identity matrix.

We are interested in the sumset of two codebooks. More precisely,  let $k_2\leq k_1\leq n$ and use $\ve m\in\mathbb F_q^{k_1}$, $\ve n'\in\mathbb F_q^{k_2}$ to denote two different messages.  We concatenate the messages of the codebook with the smaller rate as $\ve n:=
[^{\ve n'}_\ve 0] $ where $\ve 0$ is a zero vector of length $k_1-k_2$. Two codebooks are generated as
\begin{subequations}
\begin{align}
\mathcal C_1&:=\left\{\ve t :\ve t=\ve G\ve m, \mbox{ for all }\ve m\in \mathbb F_q^{k_1}\right\}\\
\mathcal C_2&:=\left\{\ve v :\ve v=\ve G\ve n=\ve G\begin{bmatrix}
\ve n' \\ 
\ve 0
\end{bmatrix}, \mbox{ for all }\ve n'\in \mathbb F_q^{k_2} \right\}
\end{align}
\label{eq:nestedcodes}
\end{subequations}
with some matrix $\ve G\in\mathbb F_q^{n\times k_1}$. Since the two codebooks are generated with the common generator matrix $\ve G$, we have $\mathcal C_2\subseteq\mathcal C_1$ and these two codebooks are called \textit{nested}. The rates of these two codebooks are $R_1:=\frac{k_1}{n}\log q, R_2:=\frac{k_2}{n}\log q$, respectively. 

From now on we will view $\mathcal C_1,\mathcal C_2$ as sets of $n$-length integer-valued vectors taking values in $\mathcal U^n$ where $\mathcal U:=\{0,\ldots,q-1\}$. The sumset of two linear codes is defined as
\begin{align}
\mathcal C_1+\mathcal C_2:=\{\ve t + \ve v : \ve t\in\mathcal C_1, \ve v\in C_2\}
\label{eq:c+c}
\end{align}
where the addition is understood as the addition in $\mathbb Z$ and is performed element-wise between the two $n$-length vectors. Hence each element in $\mathcal C_1+\mathcal C_2$ takes value in  $\mathcal W^n$ where $\mathcal W:=\{0,\ldots,2q-2\}$.    Let $T_1^n, T_2^n$  denote two random variables taking values in the code $\mathcal C_1,\mathcal C_2$ with uniform distribution, i.e.
\begin{subequations}
\begin{align}
\pp{T_1^n=\ve t}&=q^{-k_1} \mbox{ for all }\ve t\in \mathcal C_1\\
\pp{T_2^n=\ve v}&=q^{-k_2} \mbox{ for all }\ve v\in \mathcal C_2
\end{align}
\label{eq:rv_T}
\end{subequations}
The sum codewords $T_1^n+T_2^n$  is also a random vectors taking values in $\mathcal C_1+\mathcal C_2$. There is a natural distribution on $\mathcal C_1+\mathcal C_2$ induced by $T_1^n,T_2^n$, which is formally defined as follows.
\begin{definition}[Induced distribution on $\mathcal C_1+\mathcal C_2$]
Given two codebooks $\mathcal C_1,\mathcal C_2$ and assume $T_1^n,T_2^n$ are two uniformly distributed vectors defined as in (\ref{eq:rv_T}). We use $P_{\mathcal S}$ to denote the  distribution on $\mathcal C_1+\mathcal C_2$ which is induced from the distribution of $T_1^n,T_2^n$.
\label{def:P_S}
\end{definition}

 The object of interest in this paper is given in the following definition.
\begin{definition}[Typical sumset]
Let $\mathcal C_j^{(n)}, j=1,2$ be a sequence of  linear codes indexed by their dimension. Let $T_1^n,T_2^n$  be two independent random variables uniformly distributed in $\mathcal C_1^{(n)}, \mathcal C_2^{(n)}$ as in (\ref{eq:rv_T}).  A sequence of subsets $\mathcal K^{(n)}\subseteq \mathcal C_1^{(n)}+\mathcal C_2^{(n)}$ is called  \textit{typical sumsets} of $\mathcal C_1^{(n)}$, $\mathcal C_2^{(n)}$, if $T_1^n+T_2^n\in \mathcal K^{(n)}$ asymptotically almost surely, i.e., $\pp{T_1^n+T_2^n\in\mathcal  K^{(n)}}\rightarrow 1$ as $n\rightarrow\infty$.
\label{def:typical_sumset}
\end{definition}

To make notations easier, we sometimes often drop the dimension $n$ and say $\mathcal K$ is a typical sumset of $\mathcal C_1$, $\mathcal C_2$, with the understanding that a sequence of codes are considered as in Definition \ref{def:typical_sumset}. Clearly the sumset $\mathcal C_1+\mathcal C_2$ is always a typical sumset according to the definition because all possible $T_1^n+T_2^n$ must fall inside it. However we will show that for almost all linear codes,  most sum codewords $T_1^n+T_2^n$ will fall into a subset $\mathcal K$ which could be much smaller than $\mathcal C_1+\mathcal C_2$ by taking the probability distribution of $T_1^n$ and $T_2^n$ into account. In fact, we will consider a  more general case when one codebook is (possibly) shifted to a coset by a fixed vector. Assume $\mathcal C_1$ is shifted to $\mathcal C_1'$ with any fixed vector $\ve d$ as
\begin{align}
\mathcal C_1'=\mathcal C_1\oplus \ve d:=\{\ve t\oplus \ve d: \ve t\in\mathcal C_1\}
\label{eq:coset},
\end{align}
the following theorem states the main result in this section.


\begin{theorem}[Normal typical sumsets]
Let $\mathcal C_1^{(n)}, \mathcal C_2^{(n)}$ be two sequences of linear codes in $\mathbb F_q^n$ indexed by their dimension with rate $R_j:=\lim_{n\rightarrow\infty}\frac{1}{n}\log|\mathcal C_j|, j=1,2$.  For any fixed vector $\ve d\in\mathbb F_q^n$ we define $\mathcal C_1'^{(n)}:=\mathcal C_1^{(n)}\oplus \ve d$ as in (\ref{eq:coset}). Consider the case when $\mathcal C_1^{(n)}, \mathcal C_2^{(n)}$ are generated as in (\ref{eq:nestedcodes}) with the same generator matrix $\ve G$ and assume without loss of generality that $\mathcal C_2^{(n)}\subseteq\mathcal C_1^{(n)}$. If each entry of the generator matrix $\ve G$ is independent and identically chosen according to the uniform distribution $P_U$,  then asymptotically almost surely there exists a sequence of typical sumsets $\mathcal K_N^{(n)}\subseteq \mathcal C_1'^{(n)}+\mathcal C_2^{(n)}$  whose sizes satisfy
\begin{align}
|\mathcal K_N^{(n)}|&\doteq \min\left\{2^{n(R_1+R_2)}, 2^{n\left(\max\{R_1,R_2\}+D(q)\right)}\right\}\label{eq:size_sumset}\\
D(q)&:=H(U_1+U_2)-\log q
\label{eq:Dq}
\end{align}
where $U_1, U_2$ are independent random variables with the uniform distribution $P_U$ in (\ref{eq:P_U}).
Furthermore  for all $\ve w\in \mathcal K_N^{(n)}$, the induced distribution $P_{\mathcal S}$ defined in Definition \ref{def:P_S} satisfies
\begin{align}
P_{\mathcal S}(\ve w)\doteq \max\left\{2^{-n(R_1+R_2)}, 2^{-n\left(\max\{R_1,R_2\}+D(q)\right)}\right\}.
\label{eq:AEP}
\end{align}
where $P_{\mathcal S}$ is the induced probability distribution on $\mathcal C_1'^{(n)}+\mathcal C_2^{(n)}$.
\label{thm:size_sumset}
\end{theorem}
\begin{proof}
A proof of the theorem is given in Section \ref{sec:proof}. In Appendix \ref{sec:Dq} we show that $D(q)$ is an increasing function of $q$  and
\begin{align}
1/2\leq D(q) < \log \sqrt{e}\approx 0.7213
\end{align}
where the lower bound holds for $q=2$ and the upper bound is approached with $q\rightarrow\infty$.
\end{proof}

\begin{remark}
We point out that there exist linear codes which possess  (exponentially) smaller or larger typical sumsets than $\mathcal K_N$ in (\ref{eq:size_sumset}). For example $|\mathcal C_1+\mathcal C_2|$ is always larger or equal to $|\mathcal K_N|$ and we will give an example of a smaller typical sumset  in Section \ref{sec:proof}, Remark \ref{remark:small_sumset}.  To distinguish the specific typical sumset $\mathcal K_N$ in Theorem \ref{thm:size_sumset} from other possible typical sumsets, we will call $\mathcal K_N$ a \textit{normal typical sumset}. Theorem \ref{thm:size_sumset} shows that  randomly generated linear codes  a.a.s. have  a normal typical sumset $\mathcal K_N$.
\end{remark}

To help us visualize  the rather complicated expression in (\ref{eq:size_sumset}), Figure \ref{fig:K_R1R2} depicts the size of the typical sumset $\mathcal K_N$. It is also instructive to see how the typical sumset size grows if we fix the rate of one codebook and vary the rate of the other. In Figure \ref{fig:K_fixR1} we fix  $R_1$ and plot the size of the normal typical sumset for different $R_2$. Depending on the range of $R_1$, there are two cases where we have different behaviors of the typical sumset size as $R_2$ increases.   It is worthy to point out  the ``saturation" behavior on the size of the typical sumsets. For example let $R_1$ be its maximal value $\log q$ and increase $R_2$ from $0$ to $\log q$, the typical sumset size increases until $R_2$ reaches $D(q)$, but stays unchanged afterwards. It means for $R_1=\log q$ and $R_2=D(q)$, all possible sum codewords have already appeared in the typical sumset, and adding more codewords to $\mathcal C_2$ will not create new sum codewords in the typical sumset.

\begin{figure}[!htb]
\centering
\includegraphics[scale=0.5]{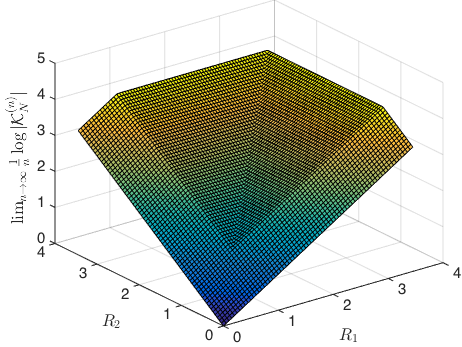}
\caption{The asymptotic size of the typical sumset $\lim_{n\rightarrow\infty}\frac{1}{n}\log|\mathcal K_N|$ in (\ref{eq:size_sumset}) as a function of $R_1, R_2$  (In this plot we set $q=11$).} 
\label{fig:K_R1R2}
\end{figure}

\begin{figure*}[htbp!]
   \centerline{\subfigure[For a fixed $R_1$ in the range $0\leq R_1\leq D(q)$, the normalized size of the typical sumset takes the form $\lim_{n\rightarrow\infty}\log |\mathcal K_N|/n=R_1+R_2$ as a function of $R_2$. ]{\includegraphics[width=2.7in]{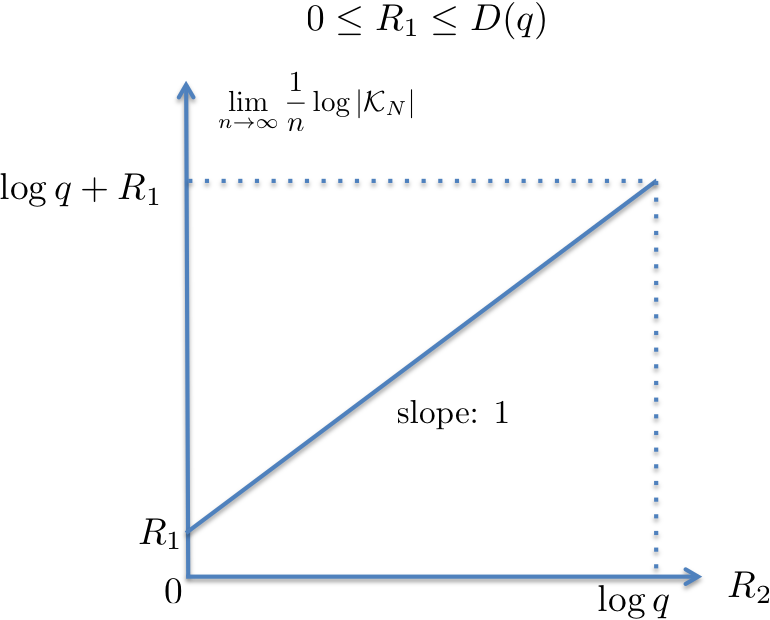}
       \label{fig:K_fixR1_small}}
     \hfil
     \subfigure[For a fixed $R_1$ in the range $D(q)\leq R_1\leq \log q$, we can identify three regimes of the growth of the typical sumset.  The piece-wise linear function $\lim_{n\rightarrow\infty}\frac{1}{n}\log |\mathcal K_N|$ is equal to  $R_1+R_2$ for $0\leq R_2< D(q)$, to $D(q)+R_1$ for $D(q)\leq  R_2<  R_1$ and to $R_2+D(q)$ for $R_1\leq  R_2\leq\log q$, as a function of $R_2$.]{\includegraphics[width=2.7in]{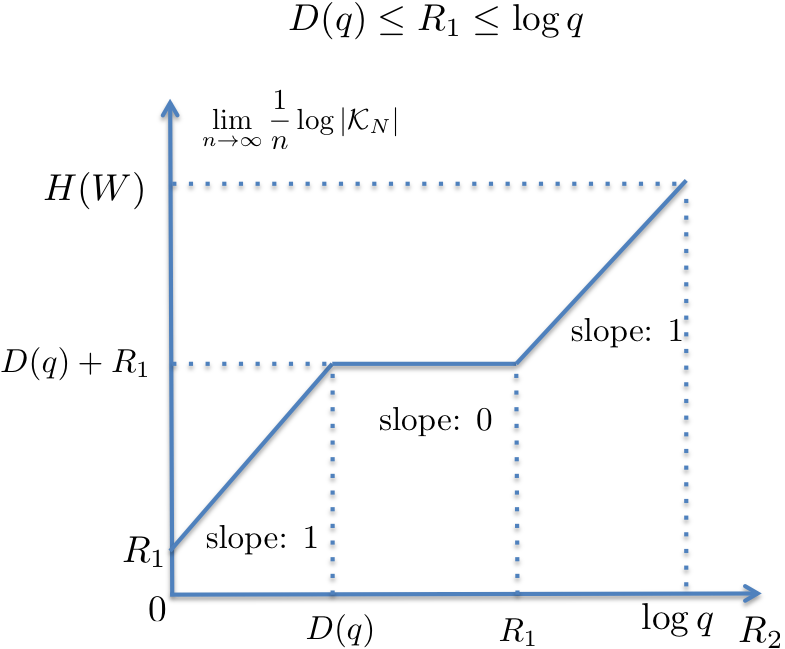}
       \label{fig:K_fixR1_large}}}
   \caption{For a fixed $R_1$, the asymptotic size of the normal typical sumset $\lim_{n\rightarrow \infty}\frac{1}{n}\log |\mathcal K_N|$ as a function of $R_2$.  }
   \label{fig:K_fixR1}
\end{figure*}

\subsection{The symmetric case}
In the case when the two codebooks are the same, i.e., $\mathcal C_1=\mathcal C_2=\mathcal C$, the size of the typical sumset is easier to describe.

\begin{corollary}[Normal typical sumsets--symmetric case]
Let $\mathcal C^{(n)}$ be a sequence linear codes indexed by their dimension in $\mathbb F_q^n$  with rate $R=\lim_{n\rightarrow\infty}\frac{1}{n}\log |\mathcal C^{(n)}|$ and let $\mathcal C'^{(n)}:=\mathcal C^{(n)}\oplus \ve d$ for any fixed $\ve d\in\mathbb F_q^n$. We assume $\mathcal C$ is generated as in (\ref{eq:linear_codes}) and each entry of the generator matrix $\ve G$ is independent and identically distributed according to the uniform distribution in $\mathbb F_q$.  Then a.a.s. there exists a sequence of typical sumsets $\mathcal K^{(n)}_N\subseteq\mathcal C^{(n)}+\mathcal C'^{(n)}$ whose sizes satisfy
\begin{align}
|\mathcal K_N^{(n)}|&\doteq
\begin{cases}
2^{2nR} &R\leq D(q)\\
2^{n(R+D(q))} &R>D(q)
\end{cases}
\label{eq:size_sumset_symmetric}
\\D(q)&:=H(U_1+U_2)-\log q.
\end{align}
where $U_1, U_2$ are independent variables with the distribution $P_U$ in  (\ref{eq:P_U}). Furthermore  for all $\ve w\in \mathcal K_N^{(n)}$, the induced distribution $P_{\mathcal S}$ defined in Definition \ref{def:P_S} satisfies
\begin{align}
P_{\mathcal S}(\ve w)\doteq
\begin{cases}
2^{-2nR} &R\leq D(q)\\
2^{-n(R+D(q))} &R>D(q)
\end{cases}
\label{eq:AEP_symmetric}
\end{align}
\label{thm:size_sumset_symmetric}
\end{corollary}
\begin{proof}
This is a consequence of Theorem \ref{thm:size_sumset} by setting $R_1=R_2=R$.  This gives
\begin{align}
|\mathcal K_N^{(n)}|\doteq\min\left\{2^{2nR}, 2^{n(R+D(q))}\right\}.
\end{align}
It is also instructive to rewrite it in the formulation stated in the corollary.
\end{proof}

For the symmetric case, Figure \ref{fig:cartoon} provides a generic plot showing the code rate $R$ vs. normalized size $\lim_{n\rightarrow\infty}\frac{1}{n}\log |\mathcal K_N^{(n)}|$ of the normal typical sumset size.  We see there exists a threshold  $D(q)$ on the rate $R$ of the code,  above or below which the normal typical sumset $\mathcal K_N$ behaves differently. 
For the low rate regime $R<D(q)$, almost every codeword pair $T_1^n, T_2^n$ gives a distinct sum codeword, hence the sumset size $|\mathcal K_N|$ is essentially $|\mathcal C|^2$. For the medium to high rate regime $R\geq D(q)$, due to the linear structure of the code, there are (exponentially) many different codeword pairs $T_1^n, T_2^n$ which give the same sum codeword, and the normal typical sumset size $|\mathcal K_N|$ grows only as  $2^{nD(q)}|\mathcal C|$ where $D(q)$ does not depend on $R$.  In this regime the code $\mathcal C$ has a typical sumset which is exponentially smaller than $\mathcal C+\mathcal C'$.  In contrast to the low dimensional case where the sum of two uniformly distributed random variables is not uniformly distributed, the sum codewords  are \textit{uniformly} distributed in the typical sumset $\mathcal K_N$ as the dimension $n$ tends to infinity,  as shown by (\ref{eq:AEP}) in Theorem \ref{thm:size_sumset}. This is reminiscent of the classical typical sequences with asymptotic equipartition property (AEP), i.e., the typical sumset occurs a.a.s. but is uniformly filled up with only a small subset of sequences. We also give a pictorial description of the sum codewords $T_1^n+T_2^n$ in Figure \ref{fig:typicalsumset}.

\begin{figure}[htbp!]
\centering
\includegraphics[scale=0.6]{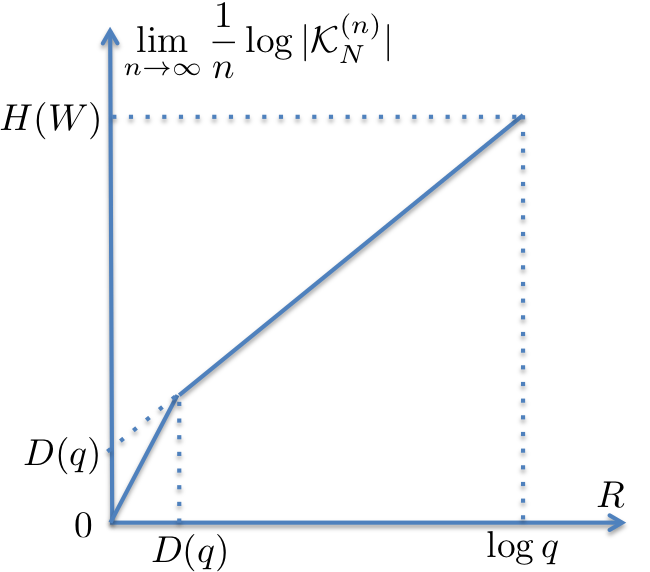}
\caption{An illustration of the size of  normal typical sumsets  of linear codes in the symmetric case.  $H(W)$ and $D(q)$ are given in (\ref{eq:H_W}) and (\ref{eq:Dq}), respectively. The piece-wise linear function has slope $2$ for low rate regime and slope $1$ for medium-to-high rate regime.}
\label{fig:cartoon}
\end{figure}




\begin{figure*}[hbt!]
\centering
\includegraphics[scale=0.55]{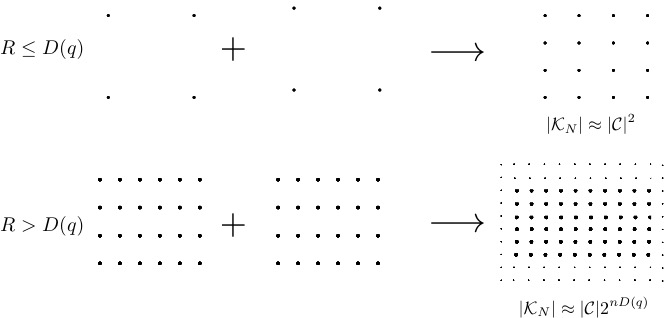}
\caption{An illustration of the sum codewords $T_1^n+T_2^n$ in the symmetric case $\mathcal C_1=\mathcal C_2=\mathcal C$. For the rate $R\leq D(q)$, each pair $(T_1^n,T_2^n)$ will give a different sum and  typical sumset $\mathcal K_N$ is essentially the same as $\mathcal C+\mathcal C$.  For rate $R>D(q)$, many pairs $(T_1^n, T_2^n)$ give the same sum codeword and the typical sumset $\mathcal K_N$ is much smaller than $\mathcal C+\mathcal C$. Interestingly in the $n$-dimensional space with $n\rightarrow\infty$, the sum codewords $T_1^n+T_2^n$ is always \textit{uniformly} distributed in the typical sumset $\mathcal K_N$ (represented by thick dots in the plot). The other sum codewords in $(\mathcal C+\mathcal C)\setminus \mathcal K_N$ (represented by the small dots) have only negligible probability.}
\label{fig:typicalsumset}
\end{figure*}


\subsection{Comparison with $|\mathcal C_1+\mathcal C_2|$}
In Section \ref{sec:intro} we emphasized the distinction between the classical sumset theory and our study of typical sumsets in a probabilistic setting. Now we compare the size of a normal typical sumset $\mathcal K_N$ of $\mathcal C_1,\mathcal C_2$ with the size of the  exact sumset $\mathcal C_1+\mathcal C_2$. Before doing this, we first introduce a useful result relating the sumsets of general linear codes with that of systematic linear codes.

\begin{lemma}[Equivalence between systematic and non-systematic codes]
Given any linear codes $\mathcal C_1, \mathcal C_2$ such that  $\mathcal C_2\subseteq\mathcal C_1$, there exist  systematic linear codes $\mathcal C_1',\mathcal C_2'$ with a one-to-one mapping $\phi: \mathcal C_1\longrightarrow \mathcal C_1', \phi: \mathcal C_2\longrightarrow\mathcal C_2'$ such that for \textit{any} pair $\ve t\in \mathcal C_1,\ve v\in \mathcal C_2$ satisfying $\ve t+\ve v=\ve s$, we have $
\phi(\ve t)+\phi(\ve v)=\phi (\ve s)$.
\label{lemma:equivalence}
\end{lemma}

\begin{proof}
A code $\mathcal C$ is said to be \textit{equivalent} (\cite[Ch. 4]{welsh_codes_1988} ) to another code $\mathcal C'$,  if  there exists a permutation $\pi$ over the set $\{1,\ldots, n\}$, such that every codeword $\ve t'$  in $\mathcal C'$ satisfies
\begin{align}
\ve t':=(\ve t'_1,\ve t'_2,\ldots, \ve t'_n)=(\ve t_{\pi(1)},\ve t_{\pi(2)},\ldots, \ve t_{\pi(n)})
\end{align}
for some $\ve t:=(\ve t_1,\ve t_2,\ldots, \ve t_n)\in \mathcal C$. It is  known that any linear code $\mathcal C$ is equivalent to some systematic linear code  (see \cite[Ch. 4.3]{welsh_codes_1988} for example).  If we assume without loss of generality that $\mathcal C_2\subseteq\mathcal C_1$,  define the mapping $\phi$ to be the permutation needed to transform the given linear code $\mathcal C_1$ to its systematic counterpart $\mathcal C_1'$. Clearly it also gives the permutation on code $\mathcal C_2$ which transforms $\mathcal C_2$ to its systematic counterpart $\mathcal C_2'$. Furthermore this permutation is a one-to-one mapping.

For two different pairs  $(\ve t, \ve v)$ and $(\tilde{\ve t}, \tilde{\ve v})$ where $\ve t,\tilde{\ve t}\in\mathcal C_1, \ve v,\tilde{\ve v}\in\mathcal C_2$ such that $\ve t+\ve v=\tilde{\ve t}+\tilde{\ve v}=\ve s$, it  holds that
\begin{align}
\phi(\ve t)+\phi(\ve v)&=(\ve t_{\pi(1)}+\ve v_{\pi(1)},\ve t_{\pi(2)}+\ve v_{\pi(2)},\ldots, \ve t_{\pi(n)}+\ve v_{\pi(n)})\\
&=(\tilde{\ve t}_{\pi(1)}+\tilde{\ve v}_{\pi(1)},\tilde{\ve t}_{\pi(2)}+\tilde{\ve v}_{\pi(2)},\ldots, \tilde{\ve t}_{\pi(n)}+\tilde{\ve v}_{\pi(n)})\\
&=\phi(\tilde{\ve t})+\phi(\tilde{\ve v})=\phi(\ve s)
\end{align}
where the second equality holds because of the assumption $\ve t+\ve v=\tilde{\ve t}+\tilde{\ve v}$ and the last equality holds because permutation is distributive with respect to entry-wise addition.
\end{proof}

This lemma shows that for any linear codes $\mathcal C_1,\mathcal C_2$ which are nested, there exists  corresponding systematic codes $\mathcal C_1',\mathcal C_2'$ whose sumset structure is exactly the same as the former. Now we can show the following simple bounds on the size of the sumset $\mathcal C_1+\mathcal C_2$.
\begin{lemma}[Simple sumset estimates]
Let $\mathcal C_1$ be an $(n,k_1)$-linear code and $\mathcal C_2$  an $(n, k_2)$-linear code over $\mathbb F_q$ such that either $\mathcal C_1\subseteq\mathcal C_2$ or $\mathcal C_2\subseteq\mathcal C_1$. The size of the sumset $\mathcal C_1+\mathcal C_2$ is upper bounded as
\begin{align}
|\mathcal C_1+\mathcal C_2|\leq q^{k_1+k_2}
\label{eq:trivial_upper}
\end{align}
and lower bounded as
\begin{align}
|\mathcal C_1+\mathcal C_2|\geq (2q-2)^{\min(k_1,k_2)}
\label{eq:AbsoSumset_lower}
\end{align}
\end{lemma}
\begin{proof}
The upper bound follows simply from the fact that $|\mathcal C_1+\mathcal C_2|\leq |\mathcal C_1||\mathcal C_2|$ for any set $\mathcal C_1,\mathcal C_2$. To establish the lower bound, Lemma \ref{lemma:equivalence} shows that for \textit{any} nested linear code $\mathcal C_1, \mathcal C_2$, we can find corresponding systematic linear codes $\mathcal C_1',\mathcal C_2'$ whose sumset size $|\mathcal C_1'+\mathcal C_2'|$ equals to $|\mathcal C_1+\mathcal C_2|$. The lower bound follows by noticing that for any systematic linear codes, the sum of the message part of the codewords already take at least $(2q-2)^{\min(k_1,k_2)}$ different values.
\end{proof}

Notice $|\mathcal K_N|$ can be smaller than  the simple lower bound given in (\ref{eq:AbsoSumset_lower}) for certain rate range. The reason is clear: some of the sum codewords $T_1^n+T_2^n$ occurs very rarely if $T_1^n$ and $T_2^n$ are chosen uniformly. Those sum codewords will be counted in the sumset $\mathcal C_1+\mathcal C_2$ but are probabilistically negligible.  
For a comparison, we consider the simple case when $k_1=k_2$ hence $\mathcal C_1$ and $\mathcal C_2$ are identical. we see the lower bound in (\ref{eq:AbsoSumset_lower}) states that
\begin{align}
|\mathcal C_1+\mathcal C_2|\geq 2^{nR_1 \log (2q-2)/\log q}.
\end{align}
Then  Eq. (\ref{eq:size_sumset}) implies that $|\mathcal K_N|$ is smaller than $|\mathcal C_1+\mathcal C_2|$ for the rate range 
\begin{align}
R>\frac{D(q)}{\log(2q-2)/\log q-1},
\end{align}
(Notice that the RHS is always larger than $D(q)$ for $q\geq 2$ but  is only meaningful if it is smaller than $\log q$). For example  $|\mathcal K_N|$ is smaller than the lower bound in (\ref{eq:AbsoSumset_lower}) for $R>2.85$ bits with $q=11$ and for $R>4.87$ bits for $q=101$.

\subsection{Entropy of sumsets}
Often we are interested in inequalities relating the entropy of two random variables $X_1, X_2$ and the entropy of their sum $X_1+X_2$. One classical result is the \textit{entropy power inequality} involving differential entropy.  Recent works including  \cite{tao_sumset_2010} \cite{kontoyiannis_sumset_2014} have established several fundamental results on this topic. For our problem,  if  codes $\mathcal C_1, \mathcal C_2$ have a normal typical sumset and 
 $T_1^n, T_2^n$ are random variables uniformly distributed in $\mathcal C_1, \mathcal C_2$ respectively, we are able to give an asymptotic relationship between $H(T_1^n), H(T_2^n)$ and $H(T_1^n+T_2^n)$.

\begin{theorem}[Entropy of sumsets]
Let $\mathcal C_1^{(n)}, \mathcal C_2^{(n)}$ be two sequences of linear codes in $\mathbb F_q^n$ with normal typical sumsets $\mathcal K_N^{(n)}$ as in Theorem \ref{thm:size_sumset}. Let $T_1^n, T_2^n$ be independent random $n$-length vectors uniformly distributed in the code $\mathcal C_1^{(n)}, \mathcal C_2^{(n)}$, respectively. In the limit $n\rightarrow\infty$  we have
\begin{align}
\lim_{n\rightarrow\infty}\frac{H(T_1^n+T_2^n)}{n}&=\min\left\{\frac{H(T^n_1)+ H(T^n_2)}{n}, \frac{\max\{H(T^n_1), H(T^n_2)\}}{n}+D(q)\right\}\\
&=\min\{R_1+R_2, \max\{R_1,R_2\}+D(q)\}
\end{align}
where as before, $D(q):=H(W)-\log q$ with $W$ distributed according to $P_W$ in (\ref{eq:P_W}). 
\label{thm:entropy}
\end{theorem}
\begin{proof}
As $T_j^n, j=1,2$ is uniformly distributed in the $(n,k_j)$-linear code $\mathcal C_j$ with rate $R_j$, we have $H(T_j^n)=nR_j$.  Theorem \ref{thm:size_sumset} shows that the distribution of the random variable $T_1^n+T_2^n$ depends on the values $R_1+R_2$ and $\max\{R_1,R_2\}+D(q)$.  We first consider the case when the $R_1+R_2$ is smaller than the latter value. Recall that $P_{\mathcal S}$ denotes the distribution on $\mathcal C_1+\mathcal C_2$ induced by $T_1^n, T_2^n$ as in Definition \ref{def:P_S},   we have
\begin{align}
H(T_1^n+T_2^n)&=-\sum_{\ve w\in \mathcal C_1+\mathcal C_2}P_{\mathcal S}(\ve w)\log P_{\mathcal S}(\ve w)\\
&\geq -\sum_{\ve w\in \mathcal K_N}P_{\mathcal S}(\ve w)\log P_{\mathcal S}(\ve w)
\end{align}
Theorem \ref{thm:size_sumset} shows that in this case for $\ve w\in \mathcal K_N$ it holds that $P_{\mathcal S}(\ve w)\leq  2^{-n(R_1+R_2-\epsilon_n)}$, hence
\begin{align}
H(T_1^n+T_2^n)&\geq -\log 2^{-n(R_1+R_2-\epsilon_n)} \sum_{\ve w\in \mathcal K_N}P_{\mathcal S}(\ve w)\\
&=n(R_1+R_2-\epsilon_n)(1-\delta_n)
\end{align}
with $\delta_n\rightarrow 0$ because $\mathcal K_N$ is a typical sumset. It follows that
\begin{align}
\lim_{n\rightarrow\infty} H(T_1^n+T_2^n)/n&\geq \lim_{n\rightarrow\infty}(R_1+R_2-\epsilon_n)(1-\delta_n)\\
&=R_1+R_2=(H(T_1)+H(T_2))/n
\end{align}
as both $\delta_n,\epsilon_n\rightarrow 0$.

On the other hand, we have 
\begin{align}
H(T_1^n+T_2^n)&=-\sum_{\ve w\in \mathcal K_N}P_{\mathcal S}(\ve w)\log P_{\mathcal S}(\ve w)-\sum_{\ve w\notin \mathcal K_N}P_{\mathcal S}(\ve w)\log P_{\mathcal S}(\ve w)
\end{align}
For $\ve w\in \mathcal K_N$ it holds $P_{\mathcal S}(\ve w)\geq 2^{-n(R_1+R_2+\epsilon_n)}$ in this case, implied by Theorem \ref{thm:size_sumset}. Hence the first term above is bounded as
\begin{align}
-\sum_{\ve w\in \mathcal K_N}P_{\mathcal S}(\ve w)\log P_{\mathcal S}(\ve w)&\leq -\log 2^{-n(R_1+R_2+\epsilon_n)}\sum_{\ve w \in \mathcal K_N}P_{\mathcal S}(\ve w)\\
&\leq n(R_1+R_2+\epsilon_n)
\end{align}
To bound the second term, using  \textit{log sum inequatliy} \cite[Lemma 3.1]{csiszar_information_2011} gives
\begin{align}
-\sum_{\ve w\notin \mathcal K_N}P_{\mathcal S}(\ve w)\log P_{\mathcal S}(\ve w)&\leq -\left(\sum_{\ve w\notin \mathcal K_N}P_{\mathcal S}(\ve w)\right) \log\frac{\sum_{\ve w\notin \mathcal K_N}P_{\mathcal S}(\ve w)}{|\overline{\mathcal K_N}|}\\
&=-P_{\mathcal S}(\overline{\mathcal K_N})\log P_{\mathcal S}(\overline{\mathcal K_N})+P_{\mathcal S}(\overline{\mathcal K_N})\log |\overline{\mathcal K_N}|
\label{eq:entropy_typical}
\end{align}
where $\overline{\mathcal K_N}$ denotes the complementary set  $(\mathcal C_1+\mathcal C_2 )\backslash\mathcal K_N$. Later in Lemma \ref{lemma:typical_matter} Eq. (\ref{eq:K_speed}) we show that 
\begin{align}
P_{\mathcal S}(\overline{\mathcal K_N})\leq 4q e^{-n\delta^2\min\{R_1,R_2\}/\log q}
\end{align}
For $n\rightarrow\infty$, the first term in (\ref{eq:entropy_typical}) approaches zero as $P_{\mathcal S}(\overline{\mathcal K_N})\rightarrow 0$. The second term is bounded as
\begin{align}
P_{\mathcal S}(\overline{\mathcal K_N})\log |\overline{\mathcal K_N}|&\leq 4q e^{-n\delta^2\min\{R_1,R_2\}/\log q}\log 2^{n(R_1+R_2)}\\
&= 4n(R_1+R_2)q e^{-n\delta^2\min\{R_1,R_2\}/\log q}
\end{align}
approaches zero as well for large enough $n$. Hence overall we have
\begin{align}
\lim_{n\rightarrow \infty}H(T_1^n+T_2^n)&\leq \lim_{n\rightarrow \infty}(R_1+R_2+\epsilon_n)+o_n(1)\\
&=n(R_1+R_2)=(H(T_1^n)+H(T_2^n))/n
\end{align}
This shows in the limit we have $H(T_1^n+T_2^n)/n\rightarrow (H(T_1^n)+H(T_2^n))/n$ for the case  $R_1+R_2\leq \max\{R_1,R_2\}+D(q) $. The other case can be proved in the same way.
\end{proof}

\section{Proof of Theorem \ref{thm:size_sumset}}\label{sec:proof}
We prove Theorem \ref{thm:size_sumset} in a few steps. Lemma \ref{lemma:equivalence} already shows that for any linear codes $\mathcal C_1,\mathcal C_2$, there exist corresponding systematic linear codes whose sumset structure is the same as the former. Hence we first focus on systematic linear codes and establish a similar result. Given two matrices $\ve Q\in\mathbb F_q^{(n-k_1)\times k_1}$ and $\ve H\in \mathbb F_q^{(k_1-k_2)\times k_2}$ with $k_1\geq k_2$, we consider two codes of the form
\begin{subequations}
\begin{align}
\mathcal C_1&=\left\{\ve t :\ve t=\begin{bmatrix}
\ve I_{k_1\times k_1}\\
\ve Q
\end{bmatrix}\ve m, \mbox{ for all }\ve m\in \mathbb F_q^{k_1}\right\}\\
\mathcal C_2&=\left\{\ve v :\ve v=
\begin{bmatrix}
\ve I_{k_1\times k_1}\\
\ve Q
\end{bmatrix}
\ve n
=
\begin{bmatrix}
\ve I_{k_1\times k_1}\\
\ve Q
\end{bmatrix}
\begin{bmatrix}
\ve n'\\
\ve H\ve n'
\end{bmatrix}, \mbox{ for all }\ve n'\in \mathbb F_q^{k_2}\right\}
\end{align}
\label{eq:systematic_proof}
\end{subequations}
where we defined $\ve n:=\begin{bmatrix}
\ve n'\\
\ve H\ve n'
\end{bmatrix}$. It is easy to see that we have $\mathcal C_2\subseteq\mathcal C_1$ in this case. Also notice that it is in general insufficient to set $\ve H$ to be the zero matrix. For example for the case $k_1=n$, letting $\ve H$ to be the zero matrix will result in a $\mathcal C_2$ whose codewords do not have parity part.

%

\begin{theorem}[Normal typical sumset - systematic linear codes]
Let $\mathcal C_1^{(n)}, \mathcal C_2^{(n)}$ be two sequences of \textit{systematic} linear codes in the form (\ref{eq:systematic_proof}) indexed by their dimension. The rates of the two codes are given by $R_j=\lim_{n\rightarrow\infty}\frac{1}{n}\log |\mathcal C_j^{(n)}|$ for $j=1,2$. For any fixed vector $\ve d\in\mathbb F_q^n$ define $\mathcal C_1'^{(n)}:=\mathcal C_1^{(n)}\oplus \ve d$ as in (\ref{eq:coset}).  If each entry of the  matrices $\ve Q, \ve H$ is independent and identically distributed according to the uniform distribution in $\mathbb F_q$, then asymptotically almost surely there exists a sequence of typical sumsets $\mathcal K_N^{(n)}\subseteq \mathcal C_1'^{(n)}+\mathcal C_2^{(n)}$ with sizes given in (\ref{eq:size_sumset}). Furthermore,  the induced probability distribution $P_{\mathcal S}$  on $\mathcal C_1'^{(n)}+\mathcal C_2^{(n)}$  satisfies (\ref{eq:AEP}). 
\label{thm:size_sumset_systematic}
\end{theorem}

\begin{remark}
There exist linear codes with a smaller typical sumset than $|\mathcal K_N|$. As an extreme example consider the sumset $\mathcal C+\mathcal C$ where a systematic $(n,k)$-linear codes $\mathcal C$ is generated with the generator matrix $[\ve I; \ve 0]$, i.e., the $\ve Q$ matrix is the zero matrix. Since the sum codewords are essentially $k$-length sequences with each entry i.i.d. with  distribution $P_W$, it is easy to see that the set of typical sequences $\mathcal A_{[W]}^k$ is actually a typical sumset for this code with size $2^{kH(W)}=2^{nRH(W)/\log q}$ where $W$ has the distribution in (\ref{eq:P_W}). This code has a typical sumset which is smaller than the normal typical sumset as demonstrated in Figure \ref{fig:small_sumset}. However this kind of codes are rare and the above theorem states that a randomly picked systematic linear code has a normal typical sumset a.a.s..
\label{remark:small_sumset}
\end{remark}

\begin{figure}[hbt!]
\centering
\includegraphics[scale=0.6]{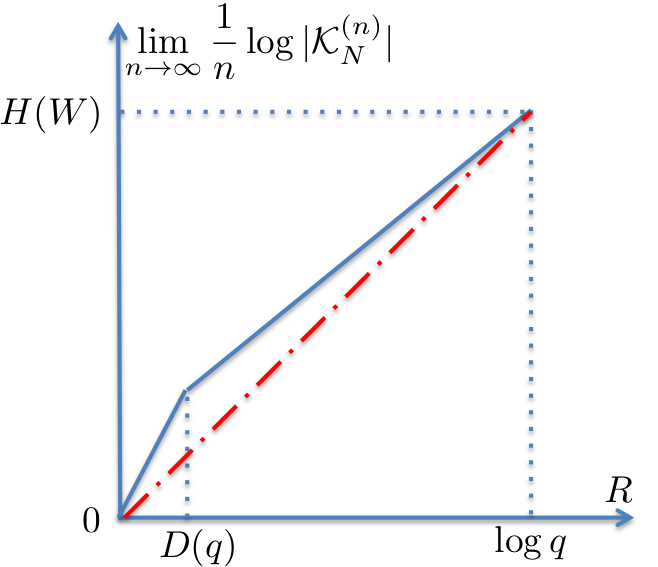}
\caption{A linear code with a typical sumset which is not normal:  the solid line shows the size of  the normal typical sumset and the dot-dashed line shows the size of a typical sumset of the example in Remark \ref{remark:small_sumset}. This code has a small typical sumset with size $2^{nR H(W)/\log q}$ but is uninteresting for the purpose of error corrections.}
\label{fig:small_sumset}
\end{figure}


We first prove Theorem \ref{thm:size_sumset_systematic}. In the following we will always assume without loss of generality that $k_1\geq k_2$. Let $\mathcal C_1$ be an $(n,k_1)$-systematic linear code and $\mathcal C_2$ be an $(n,k_2)$-systematic linear code generated using the same  generator matrix $[\ve I; \ve Q]$ as in (\ref{eq:systematic_proof}).  We fix a vector $\ve d$ and let $\mathcal C_1'=\mathcal C_1\oplus\ve d$ as in (\ref{eq:coset}).   We use $\ve d_1$ to denote the first $k_1$ entries of $\ve d$, $\ve d_2$ to denote the entries from $k_1-k_2$ to $k_1$ and $\ve d_3$ the last $n-k_1$ entries of $\ve d$.   Assume two messages $\ve m, \ve n'$ are independently and uniformly chosen from $\mathbb F_q^{k_1}, \mathbb F_q^{k_2}$, respectively,  and two codewords $\ve t\in\mathcal C_1$,  $\ve v\in\mathcal C_2$ are formed using $\ve m, \ve n'$ as in (\ref{eq:systematic_proof}).   The sum codeword of $\mathcal C_1'+\mathcal C_2$ can be written as
\begin{align}
(\ve t\oplus \ve d)+\ve v=\begin{bmatrix}
(\ve m_1\oplus \ve d_1)+\ve n'\\
(\ve m_2\oplus \ve d_2)+\ve H \ve n'\\
(\ve Q\ve m\oplus \ve d_3)+\ve Q\ve n
\end{bmatrix}
:=\begin{bmatrix}
\sf s(\ve m, \ve n')\\
\sf p_1(\ve m, \ve n')\\
\sf p_2(\ve m, \ve n')
\end{bmatrix}
\label{eq:sum_systematic}
\end{align}
where we use $\ve m_1$ to denote the first $k_2$ entries of $\ve m$ and $\ve m_2$ to denote its remaining entries. We use $\sf s(\ve m,\ve n')$ to denote the first $k_2$ entries of the sum codewords. We also use $\sf p_1(\ve m, \ve n')$ and $\sf p_2(\ve m, \ve n')$ to denote the entries of the sum codewords with indices ranging from $k_2$ to $k_1$, and with indices ranging from $k_1$ to $n$, respectively. In the sequel we will refer to $\sf s(\ve m,\ve n')$ and $\sf p_1(\ve m, \ve n'), \sf p_2(\ve m, \ve n')$ defined above as the \textit{information-sum} and \textit{parity-sum}, respectively. We shall omit their dependence on $\ve m,\ve n'$ and use $\sf s, \sf p_1, \sf p_2$ if it is clear in the context.

We choose  $\mathcal K_N$ to be the set which  contains  sum codewords whose information-sum $\ve s$  is typical, that is
\begin{align}
\mathcal K_N:=\left\{(\ve t\oplus \ve d)+\ve v \middle| (\ve t\oplus \ve d)+\ve v=\begin{bmatrix}\sf s\\ {\sf p}_1 \\ {\sf p}_2\end{bmatrix} \text{ where }   {\sf s}\in \mathcal A^{(k_2)}_{[W]} \right\}
\label{eq:K_normal}
\end{align}
with ${\sf s}, {\sf p}_1, {\sf p}_2$ defined in (\ref{eq:sum_systematic}) and $W$ defined in (\ref{eq:P_W}). For all pairs of codewords $(\ve t, \ve v)$ whose information-sum  equals to a common value $\ve s$, we define the set of all possible parity-sums as
\begin{align}
\mathcal P_{\ve Q, \ve H}(\ve s):=\left\{  \begin{bmatrix}
(\ve m_2\oplus \ve d_2)+\ve H \ve n'\\
(\ve Q\ve m\oplus \ve d_3)+\ve Q\ve n
\end{bmatrix}:\ve m\in \mathbb F_q^{k_1}, \ve n' \in\mathbb F_q^{k_2}\mbox{ such that } \sf s(\ve m,\ve n')=\ve s \right\}.
\label{def:p}
\end{align}
with $\ve n:=\begin{bmatrix}
\ve n'\\
\ve H\ve n'
\end{bmatrix}$. To facilitate our analysis, we further decompose the above set in the following way. When the information-sum is fixed to be $\sf s(\ve m,\ve n')=\ve s $, we define the set of  possible parity-sums $\sf p_1$ as
\begin{align}
\mathcal P_{1, \ve H}(\ve s):=\left\{  
(\ve m_2\oplus \ve d_2)+\ve H \ve n': \ve m\in \mathbb F_q^{k_1}, \ve n' \in\mathbb F_q^{k_2}\mbox{ such that } \sf s(\ve m,\ve n')=\ve s  \right\}.
\label{def:p1}
\end{align}
When the information-sum  $\sf s$ is fixed to be $\sf s(\ve m,\ve n')=\ve s $ and the parity-sum $\sf p_1$ is fixed to be ${\sf p}_1(\ve m,\ve n')=\ve p_1  $, we also define the set of  possible parity-sums ${\sf p}_2$ as
\begin{align}
\mathcal P_{2,\ve Q, \ve H}(\ve s, \ve p_1):=\left\{  
(\ve Q\ve m\oplus \ve d_3)+\ve Q\ve n: \ve m\in \mathbb F_q^{k_1}, \ve n' \in\mathbb F_q^{k_2}\mbox{ such that } {\sf s}(\ve m,\ve n')=\ve s , {\sf p}_1(\ve m,\ve n')=\ve p_1 \right\}.
\label{def:p2}
\end{align}

Notice we have the following relationship between the cardinality of the above three sets
\begin{align}
|\mathcal P_{\ve Q,\ve H}(\ve s)|=\sum_{\ve p_1\in\mathcal P_{1,\ve H}(\ve s)} |\mathcal P_{2,\ve Q,\ve H}(\ve s, \ve p_1)|
\label{eq:relation_paritysum}
\end{align}

In the following lemma we show that the set $|\mathcal K_N|$ defined in (\ref{eq:K_normal}) is indeed a typical sumset. We also give a simple estimate on its size.

\begin{lemma}[The typical sumset $|\mathcal K_N|$]
Let $\mathcal C_1$ be an $(n,k_1)$-systematic linear code and $\mathcal C_2$  an $(n,k_2)$-systematic linear code ($k_1\geq k_2$) which are generated as in (\ref{eq:systematic_proof}). Let $\mathcal C_1'=\mathcal C_1\oplus \ve d$ for any fixed $\ve d$ and $T_1^n,T_2^n$ be two random variables uniformly distributed in  $\mathcal C_1', \mathcal C_2$, respectively. We have
\begin{align}
\pp{T_1^n+T_2^n\in \mathcal K_N}\rightarrow 1 \mbox{ as }n\rightarrow\infty
\end{align}
with $\mathcal K_N$ defined in (\ref{eq:K_normal}). Furthermore we have
\begin{align*}
|\mathcal K_N|=\sum_{\ve s\in \mathcal A^{(k_2)}_{[W]}}|\mathcal P_{\ve Q, \ve H}(\ve s)|
\end{align*}
with $W$ defined in  (\ref{eq:P_W}).
\label{lemma:typical_matter}
\end{lemma}
\begin{proof}
Recall that we defined $\mathcal K_N$ in  (\ref{eq:K_normal}) to be the set containing all sum codewords whose information-sum $\sf s$ satisfies the property that  $\sf s$ is a typical sequence in $\mathcal A^{(k_2)}_{[W]}$. As shown in (\ref{eq:sum_systematic}) $\sf s=\ve m_1\oplus\ve d_1+\ve n'$ where  $\ve m_1$ and $\ve u'$ are independent vectors and are uniformly distributed in $\mathcal U^{k_2}$, then for any fixed $\ve d$, the first $k_2$ entries of $T_1^n+T_2^n$ is in fact an i.i.d. sequence distributed according to $P_W$, thanks to the systematic form of the codes.

Let  $S^{k_2}$ denote a $k_2$-length random vector with each entry i.i.d. according to $P_W$. We have
\begin{align}
\pp{T_1^n+T_2^n\in \mathcal K_N}&=\pp{S^{k_2}\in \mathcal A^{(k_2)}_{[W]}}\\
&\geq (1-2|\mathcal W|e^{-2k_2\delta^2})\\
&> 1-4qe^{-n(2\delta^2R_2/\log q)}
\label{eq:K_speed}
\end{align}
where the first inequality follows from the property of  typical sequences in Lemma \ref{lemma:typical_sequence}. Choose $\delta$ such that $n\delta^2\rightarrow \infty$, we have that $T_1^n+T_2^n\in \mathcal K_N$ a.a.s. for $n$ large enough and $R_2>0$. This shows $\mathcal K_N$ is indeed a typical sumset.  The claim  on the size of $\mathcal K_N$ follows by the definition of $\mathcal K_N$ and $\mathcal P_{\ve Q, \ve H}(\ve s)$.  
\end{proof}

The above lemma shows that we only need to focus on the message pairs $(\ve m, \ve n')$ if the information-sum ${\sf s}(\ve m, \ve n')$ is a typical sequence $\ve s\in\mathcal A_{[W]}^{(k_2)}$ as shown in (\ref{eq:K_normal}). For a given information-sum $\ve s$, we have the following characterization.

\begin{lemma}[Message pairs with given $\sf s$]
Let $\ve m, \ve n'$ be two vectors in $\mathbb F_q^{k_1}$ and $\mathbb F_q^{k_2}$ respectively. Two codes $\mathcal C_1,\mathcal C_2$ with rate $R_1, R_2$ are generated as in (\ref{eq:systematic_proof}) and their sum codewords are of the form in (\ref{eq:sum_systematic}).  There are $L$ pairs of $(\ve m,\ve n')$ satisfying $\sf s(\ve m,\ve n')=\ve s$ for some $\ve s\in\mathcal A_{[W]}^{(k_2)}$ with
\begin{align*}
L\doteq 2^{n(R_1+R_2-R_2H(W)/\log q)}
\end{align*}
\label{lemma:information_sum}
\end{lemma}
\begin{proof}
Recall that for any fixed $\ve d$, we defined ${\sf s}(\ve m,\ve n'):=(\ve m_1\oplus \ve d_1) +\ve n'$ where $\ve m_1, \ve d_1$ is given in (\ref{eq:sum_systematic}). For a given value $\ve s_{i}\in\mathcal W$, we can write out all possible $(\ve m_{1,i}\oplus\ve d_{1,i}, \ve n_{i})$ summing up to $\ve s_i$ explicitly:
\begin{align*}
\ve s_{i}: &(\ve m_{1,i}\oplus\ve d_{1,i},\ve n'_{i}) \mbox{ such that }\ve m_{1,i}\oplus\ve d_{1,i}+\ve n'_{i}=\ve s_{i}, i=1,\ldots,k_2 \\
0: &(0,0)\\
1: &(0,1), (1,0)\\
2: &(1,1), (2,0), (0,2)\\
3: &(0,3), (3,0),(1,2), (2,1)\\
&\vdots\\
q-1: &(0,q-1),(q-1,0), (1,q-2),(q-2,1),\ldots, ((q-1)/2,(q-1)/2)\\
&\vdots\\
2q-3: &(q-1,q-2), (q-2,q-1)\\
2q-2: &(q-1,q-1)
\end{align*}
We can show that  the number of different pairs $(\ve m_1\oplus \ve d_1,\ve n')$ satisfying $\ve m_1\oplus\ve d_1+\ve n'=\ve s$ is
\begin{align}
&2^{(2/q^2+o(1))k_2}3^{(3/q^2+o(1))k_2}\ldots q^{(q/q^2+o(1))k_2}(q-1)^{((q-1)/q^2+o(1))k_2}\ldots2^{(2/q^2+o(1))k_2}\\
&=\prod_{a=1}^q a^{(a/q^2+o(1))k_2} \prod_{a=1}^{q-1} a^{(a/q^2+o(1))k_2} \\
&=2^{k_2(\log q-D(q)+o(1))}
\end{align}
To see why this is the case, recall that since $\ve s$ is a typical sequence in $\mathcal A_{[W]}^{(k_2)}$,  there are for example $(2/q^2+o(1))k_2$ entries in $\ve s$ taking value $1$, as implied by the definition of typical sequences in (\ref{eq:typical_def}) and the distribution  $P_W$. The pair $(\ve m_{1,i}\oplus\ve d_{1,i}, \ve n'_{i})$ can take value $(1,0)$ or $(0,1)$ in these entries. Hence there are $2^{(2/q^2+o(1))k_2}$ different choices on the pair $(\ve m_1\oplus\ve d_1, \ve n')$ for those entries. The same argument goes for other entries taking values $2,\ldots, 2q-2$ using the number of possible values of $(\ve m_{1,i}\oplus\ve d_{1,i},\ve n'_{i})$ shown in the above list. Furthermore since there are $q^{k_1-k_2}$ possible $\ve m_2$ for each of the $(\ve m_1\oplus \ve d_1, \ve n')$,  the number of $(\ve m,\ve n')$ giving ${\sf s}(\ve m, \ve n')=\ve s$ is
\begin{align*}
L=2^{k_2(\log q-D(q)+o(1))}\cdot q^{k_1-k_2}\doteq 2^{n(R_1+R_2-R_2H(W)/\log q)}
\end{align*}
which proves the claim.
\end{proof}

In the following lemmas we will give the estimates on the size of parity-sums.
\begin{lemma}[Estimates of $|\mathcal P_1|$]
Let $\ve m, \ve n'$ be two independent random vectors which are uniformly distributed in $\mathbb F_q^{k_1}$ and  $\mathbb F_q^{k_2}$, respectively.  For the pairs $(\ve m,\ve n')$  satisfying ${\sf s}(\ve m,\ve n')=\ve s$ for some $\ve s\in\mathcal A_{[W]}^{(k_2)}$,  let $\mathcal P_1(\ve s)$ denote the random set formed in (\ref{def:p1}), where each entry of $\ve H$ is i.i.d. according to the uniform distribution in $\mathbb F_q$. Then asymptotically almost surely it holds that
\begin{align*}
|\mathcal P_1(\ve s)|\doteq 2^{n(R_1+R_2-R_2H(W)/\log q)}
\end{align*}
if  $R_2\leq R_1D(q)/\log q$, and 
\begin{align*}
|\mathcal P_1(\ve s)|\doteq 2^{n(R_1-R_2)H(W)/\log q}
\end{align*}
if $R_2\geq R_1D(q)/\log  q$.  
\label{lemma:size_p1}
\end{lemma}
\begin{proof}
We will bound the possible number of different parity-sum ${\sf p}_1$ given the condition that ${\sf s}(\ve m, \ve n')=\ve s$ for some $\ve s\in\mathcal A_{[W]}^{(k_2)}$. It is shown in Appendix \ref{appendix:parity_sum}, Lemma \ref{lemma:check_distr} that  each entry of the parity sum ${\sf p}_1$ is i.i.d. according to $P_W$ hence the probability that the parity-sum ${\sf p}_1$ being atypical is negligible. For a given typical vector $\ve p\in\mathcal A^{(k_1-k_2)}_{[W]}$,  define the random variable $Z_1(\ve p)$ to be the number of pairs $(\ve m, \ve n')$ whose parity sum $\sf p_1$ is equal to $\ve p$. In other words, define the random set
\begin{align*}
\mathcal Z_1(\ve p):=\{(\ve m,\ve n'): \ve m_2\oplus \ve d_2+\ve H\ve n'=\ve p\}
\end{align*}
where each entry of $\ve H$ is chosen uniformly at random from $\mathbb F_q$, the random variable $Z_1(\ve p)$ is defined as $Z_1(\ve p):=|\mathcal Z_1(\ve p)|$. In Appendix \ref{app:Z_1} we show that if $\sf s(\ve m,\ve n')=\ve s$ for some $\ve s\in\mathcal A_{[W]}^{(k_2)}$ and with randomly chosen $\ve H$,  the   conditional expectation and  variance of $Z_1(\ve p)$ for a typical sequence $\ve p\in\mathcal A^{(k_1-k_2)}_{[W]}$ is  bounded as
\begin{align}
2^{n\left(R_2-R_1D(q)/\log q -\epsilon_n \right)}\leq \E{Z_1(\ve p)|\sf s(\ve m,\ve n')=\ve s}\leq 2^{n\left(R_2-R_1D(q)/\log q +\epsilon_n \right)}
\label{eq:bounds_E_Z1}
\end{align}
for some $\epsilon_n\rightarrow 0$. For any fixed $\ve p\in\mathcal A^{(k_1-k_2)}_{[W]}$,  Markov inequality shows that
\begin{align}
\pp{Z_1(\ve p)>1)|\sf s(\ve m,\ve n')=\ve s}&\leq \E{Z_1(\ve p)|\sf s(\ve m,\ve n')=\ve s}\leq  2^{n\left(R_2-R_1D(q)/\log q +\epsilon_n \right)}
\end{align}
In the case when $R_2\leq R_1D(q)/\log q-2\epsilon_n$, we have $\pp{Z_1(\ve p)>1|{\sf s}(\ve m,\ve n')=\ve s}\leq 2^{-n\epsilon_n}$ which can be made arbitrarily small for large enough $n$ if choose $\epsilon_n$ such that $n\epsilon_n\rightarrow\infty$. As $Z_1(\ve p)$ denotes the number of pairs $(\ve m, \ve n')$ which give a parity-sum as $\sf p_1(\ve m,\ve n')=\ve p$, this means a.a.s. any typical sequence $\ve p$ can be formed by at most one pair $(\ve m, \ve n')$. In other words, every pair will form a distinct $\sf p_1(\ve m, \ve n')$ a.a.s. hence the number of distinct ${\sf p}_1$ equals to the number of pairs $(\ve m, \ve n')$ satisfying ${\sf s}(\ve m,\ve n')=\ve s$, which is given by $L$ in Lemma \ref{lemma:information_sum}. This proves the first claim by letting $\epsilon_n$ go to zero.

In the case when $R_2\geq  R_1D(q)/\log q$, we show that the number of different $\sf p_1$ is concentrated around $2^{n(R_1-R_2)H(W)/\log q}$.  For some $\epsilon'_n>0$ depending on $n$, by conditional Chebyshev inequality (see \cite[Ch. 23.4]{fristedt_modern_1997} for example) we have
\begin{align}
\pp{|Z_1(\ve p)-\E{Z_1(\ve p)}|\geq 2^{\frac{n}{2}\left(R_2-R_1D(q)/\log q +\epsilon'_n \right)}|{\sf s}(\ve m,\ve n')=\ve s}&\leq \frac{\var{Z_1(\ve p)|{\sf s}(\ve m,\ve n')=\ve s}}{2^{2\cdot\frac{n}{2}\left(R_2-R_1D(q)/\log q +\epsilon'_n \right)}}\\
&\leq \frac{\E{Z_1(\ve p)|{\sf s}(\ve m,\ve n')=\ve s}}{2^{2\cdot\frac{n}{2}\left(R_2-R_1D(q)/\log q +\epsilon'_n \right)}}\\
&\leq 2^{-n(\epsilon'_n-\epsilon_n)}
\end{align}
where we used the inequality $\var{Z_1(\ve p)|{\sf s}(\ve m,\ve n')=\ve s}\leq \E{Z_1(\ve p)|{\sf s}(\ve m,\ve n')=\ve s}$ proved in Appendix \ref{app:Z_1}. If we choose $\epsilon'_n>\epsilon_n$ and $n$ such that $n(\epsilon'_n-\epsilon_n)\rightarrow\infty$ and $\epsilon'_n\rightarrow 0$ (this is possible because $\epsilon_n\rightarrow 0$), then under the condition that ${\sf s}(\ve m,\ve n')=\ve s$, $Z_1(\ve p)$ a.a.s. satisfies
\begin{align}
\E{Z_1(\ve p)|\sf s(\ve m,\ve n')=\ve s}-2^{\frac{n}{2}(R_2-R_1D(q)/\log q+\epsilon'_n)}\leq Z_1(\ve p)\leq   \E{Z_1(\ve p)|\sf s(\ve m,\ve n')=\ve s}+2^{\frac{n}{2}(R_2-R_1D(q)/\log q+\epsilon'_n)}
\label{eq:Z1_estimates}
\end{align}
Furthermore we have the following identity regarding the total number of pairs $(\ve m,\ve n')$ satisfying ${\sf s}(\ve m,\ve n')=\ve s$:
\begin{align}
\sum_{\ve p\in \mathcal P_1(\ve s) }Z_1(\ve p)=L
\label{eq:total_number_L}
\end{align}
where $L$ is given in Lemma \ref{lemma:information_sum}. Combining (\ref{eq:Z1_estimates}) and (\ref{eq:total_number_L}),  the following estimates hold a.a.s.
\begin{align}
\frac{L}{\E{Z_1(\ve p)}+2^{\frac{n}{2}(R_2-R_1D(q)/\log q +\epsilon'_n )}}\leq |\mathcal P_1(\ve s)|\leq \frac{L}{\E{Z_1(\ve p)}-2^{\frac{n}{2}( R_2-R_1D(q)/\log q +\epsilon'_n )}}
\end{align}
Using the bounds on $\E{Z_1(\ve p)|\sf s(\ve m,\ve n')=\ve s}$ in (\ref{eq:bounds_E_Z1}) and Lemma \ref{lemma:information_sum},   $\mathcal P_1(\ve s)$ can be further bounded a.a.s. as
\begin{align}
\frac{2^{n((R_1-R_2)H(W)/\log q+o(1))}}{1+2^{-\frac{n}{2}(R_2-R_1D(q)/\log q+2\epsilon_n-\epsilon_n')}}\leq |\mathcal P_1(\ve s)|\leq \frac{2^{n((R_1-R_2)H(W)/\log q+o(1))}}{1-2^{-\frac{n}{2}(R_2-R_1D(q)/\log q-2\epsilon_n-\epsilon_n')}}
\end{align}
By the assumption that $R_2\geq R_1D(q)/\log q$, we can let  $R_2=R_1D(q)/\log q+\sigma_n$ for some $\sigma_n\rightarrow 0$. The two terms in the denumerators of the above expression can be written as
\begin{align}
2^{-\frac{n}{2}(R_2-R_1D(q)/\log q+2\epsilon_n-\epsilon_n')}&=2^{-\frac{n}{2}(\sigma+2\epsilon_n-\epsilon')}\\
 2^{-\frac{n}{2}(R_2-R_1D(q)/\log q-2\epsilon_n-\epsilon_n')}&=2^{-\frac{n}{2}(\sigma-2\epsilon_n-\epsilon_n')}
\end{align}
and both terms approaches $0$ if $\sigma_n>2\epsilon_n+\epsilon_n'$. Since both $\epsilon_n$ and $\epsilon_n'$ are chosen to approach $0$, we can also let $\sigma_n$  approach $0$. This proves that for $R_2\geq R_1D(q)/\log q$ and $n$ large enough we have a.a.s.
\begin{align} 
\frac{2^{n((R_1-R_2)H(W)/\log q+o(1))}}{1+o_n(1)}\leq |\mathcal P_1(\ve s)|\leq \frac{2^{n((R_1-R_2)H(W)/\log q+o(1))}}{1-o_n(1)}
\end{align}
or equivalently $\mathcal P_1(\ve s)\doteq 2^{n((R_1-R_2)H(W)/\log q)}$ a.a.s. if $n$ is sufficiently large. 
\end{proof}

Now we will determine the size of the parity-sums $\mathcal P_2$. The following lemma gives the key property of the parity-sum $\sf p_2$.

\begin{lemma}(Key property of parity-sum ${\sf p}_2$)
Let $\ve m, \ve n'$ be two independent random vectors which are uniformly distributed in $\mathbb F_q^{k_1}$ and  $\mathbb F_q^{k_2}$, respectively. Let $\ve H\in\mathbb F_q^{(k_1-k_2)\times k_2}$ and $\ve Q\in\mathbb F_q^{k_1\times n}$ be two matrices and $\ve d_1\in\mathbb F_q^{k_2},\ve d_2\in\mathbb F_q^{k_1-k_2}, \ve d_3\in\mathbb F_q^{n-k_1}$ some fixed vectors. We consider  all  pairs $(\ve m,\ve n')$  which satisfy the condition
\begin{subequations}
\begin{align}
{\sf s}(\ve m,\ve n')&=\ve m_1\oplus \ve d_1+\ve n'=\ve s\\
{\sf p}_1(\ve m,\ve n')&=\ve m_2\oplus \ve d_2+\ve H\ve n'=\ve p_1
\end{align}
\label{eq:same_s_same_p1}
\end{subequations}
for some $\ve s\in\mathcal W^{k_2}$ and $\ve p_1\in\mathcal W^{k_1-k_2}$. Furthermore, let ${\sf p}_{2,i}$ denote the $i$-th entry  of the parity sum ${\sf p}_2(\ve m, \ve n'):=\ve Q\ve m\oplus \ve d_3+\ve Q\ve n$ with $\ve n:=\begin{bmatrix}
\ve n'\\
\ve H\ve n'
\end{bmatrix}$. Then for all pairs $(\ve m, \ve n')$ satisfying (\ref{eq:same_s_same_p1}) and any matrices $\ve Q, \ve H$, we have
\begin{align*}
{\sf p}_{2,i}(\ve m,\ve n')\in\{a, a+q\} \mbox{ with some }a\in[0:q-1] \mbox{ for all } i \in[1:n-k_1]
\end{align*}
Equivalently, define a subset $\mathcal F(\ve a)$ in $\mathcal W^{n-k_1}$ with a vector $\ve a\in\mathcal U^{n-k_1}$ as
\begin{align}
\mathcal F(\ve a):=\{\ve p: \ve p_i\in\{a_i,a_i+q\}, i\in[1:n-k_1]\},
\end{align}
we always have
\begin{align}
\mathcal P_{2,\ve Q, \ve H}(\ve s,\ve p_1)\subseteq \mathcal F(\ve a)
\end{align}
with $\mathcal P_{2,\ve Q, \ve H}(\ve s,\ve p_1)$ defined in (\ref{def:p2}) for some $\ve a\in \mathcal U^{n-k_1}$ depending  on $\ve s, \ve p_1, \ve d, \ve H$ and $\ve Q$.
\label{lemma:2values}
\end{lemma}
\begin{proof}
We  rewrite the sum
\begin{align}
\ve m_1\oplus\ve d_1+\ve n'&=\ve m_1+\ve d_1+f_q(\ve m_1, \ve d_1)+\ve n'\\
\ve m_2\oplus \ve d_2+\ve H\ve n'&=\ve m_2+\ve d_2+f_q(\ve m_2, \ve d_2)+\ve H\ve n'
\end{align}
where the function $f_q: \mathcal U^k\times \mathcal U^k\rightarrow \mathcal U^k$ returns a vector of the same length as  inputs, and its $i$-th entry is given as
\begin{align}
f_q(\ve a,\ve b)_i=
\begin{cases}
q &\text{if } \ve a_i+\ve b_i\geq q\\
0 &\text{otherwise}
\end{cases}
\end{align}
Also notice that we can always write the product $\ve a^T\ve b$ in the finite field $\mathbb F_q^k$ as  $\ve a^T\ve b=\langle \ve a, \ve b\rangle +qn$ for some integer $n$ where  $\langle \ve a,\ve b\rangle$  denotes the inner product of two vectors in $\mathbb R^k$. Use $\ve Q_i$ to denote the $i$-th column of $\ve Q$, and use $\hat{\ve Q}_{i}$ to denote the first $k_2$ entries of $\ve Q_i$ and  $\check{\ve Q}_i$ to denote the remaining $k_1-k_2$ entries of $\ve Q_i$, we can rewrite the $i$-th entry of parity sum ${\sf p}_2$ as
\begin{align}
{\sf p}_{2,i}(\ve m,\ve n')&=\ve Q_i^T\ve m\oplus\ve d_{3,i}+\ve Q_i^T\ve n\\
&=\langle \ve Q_i, \ve m\rangle +qn_1+\ve d_{3,i}+f_q(\ve Q_i^T\ve m, \ve d_{3,i})+\langle \ve Q_i, \ve n\rangle+qn_2\\
&=\langle \hat{\ve Q}_i, \ve m_1\rangle + \langle \check{\ve Q}_i, \ve m_2\rangle+qn_1+\ve d_{3,i}+f_q(\ve Q_i^T\ve m, \ve d_{3,i})+\langle \hat{\ve Q}_i, \ve n'\rangle + \langle \check{\ve Q}_i, \ve H\ve n'\rangle +qn_2\\
&= \langle \hat{\ve Q}_i, \ve m_1+\ve n'\rangle +\langle \check{\ve Q}_i, \ve m_2+\ve H\ve n'\rangle +f_q(\ve Q_i^T\ve m, \ve d_{3,i})+\ve d_{3,i}+q(n_1+n_2)\\
&\stackrel{(a)}{=}\langle \hat{\ve Q}_i, \ve s-\ve d_1-f_q(\ve m_1,\ve d_1)\rangle +\langle \check{\ve Q}_i, \ve p_1-\ve d_2-f_q(\ve m_2,\ve d_2)\rangle+\ve d_{3,i}+f_q(\ve Q^T_i\ve m,\ve d_{3,i})+q(n_1+n_2)\\
&=\langle \hat{\ve Q}_i, \ve s -\ve d_1\rangle-\langle \hat{\ve Q}_i,f_q(\ve m_1,\ve d_1)\rangle+\langle \check{\ve Q}_i, \ve p_1 -\ve d_2\rangle-\langle \check{\ve Q}_i,f_q(\ve m_2,\ve d_2)\rangle+ \ve d_{3,i}\\
&+f_q(\ve Q^T_i\ve m,\ve d_{3,i})+q(n_1+n_2)
\end{align}
In step $(a)$ we used the assumption that
\begin{align*}
\ve m_1+\ve d_1+\ve n'+f_q(\ve m_1,\ve d_1)&=\ve s\\
\ve m_2+\ve d_2+\ve H\ve n'+f_q(\ve m_2,\ve d_2)&=\ve p_1
\end{align*}
Furthermore $\langle \hat{\ve Q}_i,f_q(\ve m_1,\ve d_1)\rangle=\sum_{j=1}^n\hat{\ve Q}_{i,j}f_q(\ve m_1,\ve d_1)_j$ and since $f_q(\ve m_1,\ve d_1)_j$ is either $q$ or $0$, we have $\langle \hat{\ve Q}_i,f_q(\ve m_1,\ve d_1)\rangle=n_3q$ for some integer $n_3$. Similarly we have $\langle \check{\ve Q}_i,f_q(\ve m_2,\ve d_2)\rangle=qn_4$ for some integer $n_4$ and  $f_q(\ve Q^T_i\ve m,\ve d_{3,i})=n_5q$ where $n_5$ is either $0$ or $1$. This leads to the observation that
\begin{align}
{\sf  p}_{2,i}(\ve m,\ve n')&=\langle \hat{\ve Q}_i, \ve s -\ve d_1\rangle+ \langle \check{\ve Q}_i, \ve p_1 -\ve d_2\rangle+\ve d_{3,i}+q(n_1+n_2-n_3-n_4+n_5)\\
&= a+q(n_1+n_2-n_3+n_4+n_5+n_6)\\
&:= a+qn'
\end{align}
where in the penultimate step we write  $\langle \hat{\ve Q}_i, \ve s -\ve d_1\rangle+ \langle \check{\ve Q}_i, \ve p_1 -\ve d_2\rangle+\ve d_{3,i}=a+qn_6$ for some $a\in[0:q-1]$ and integer $n_5$. On the other hand we know ${\sf p}_{2,i}(\ve m,\ve n')$ only takes value in $[0:2q-2]$, the above expression  implies ${\sf p}_{2,i}(\ve m,\ve n')$ can only equal to $a$ or $a+q$ for some $a\in[0,q-1]$, namely $n'$ can only equal to $0$ or $1$, irrespective of which pair $(\ve m, \ve n')$ is considered. In particular if $a=q-1$, we must have $n'=0$ and $\ve p_i=q-1$. We can use the same argument for all entries ${\sf p}_{2,i}(\ve m,\ve n'), i=1,\ldots, n-k_1$ and show that the entry ${\sf p}_{2,i}(\ve m,\ve n')$ can take at most two different values for any pair $(\ve m, \ve n')$ satisfying (\ref{eq:same_s_same_p1}).  Since there are $q^{n-k_1}$ different choices of $\ve a$, we can partition the whole space $\mathcal W^{n-k_1}$ into $q^{n-k_1}$ disjoint subsets $\mathcal F(\ve a)$. For any $\ve Q, \ve H$, fix the information sum $\sf s$ to be $\ve s$ and parity sum $\sf p_1$ to be $\ve p_1$, all parity-sums $\mathcal  P_{2,\ve Q,\ve H}(\ve s,\ve p_1)$ defined in (\ref{def:p2}) are confined in a subset $\mathcal F(\ve a)$. 
\end{proof}

 To lighten the notation, for  given $\ve s, \ve p_1$ we define 
\begin{align}
F(\ve a):=\{\mathcal P_{2,\ve Q,\ve H}(\ve s, \ve p_1)\subseteq \mathcal F(\ve a)\}
\end{align}
to denote the event when all parity-sums are contained in the set $\mathcal F(\ve a)$.


\begin{lemma}[Estimates of $|\mathcal P_2|$]
Let $\ve m, \ve n'$ be two independent random vectors which are uniformly distributed in $\mathbb F_q^{k_1}$ and  $\mathbb F_q^{k_2}$, respectively.  For the pairs $(\ve m,\ve n')$  satisfying ${\sf s}(\ve m,\ve n')=\ve s$ and ${\sf p}_1(\ve m,\ve n')=\ve p_1$ for some $\ve s\in\mathcal A_{[W]}^{(k_2)}$ and $\ve p_1\in\mathcal A_{[W]}^{(k_1-k_2)}$,  let $\mathcal P_2(\ve s, \ve p_1)$ denote the random set of parity-sum ${\sf p}_2$ formed in (\ref{def:p2}), where each entry of $\ve H$ and $\ve Q$ is chosen i.i.d. uniformly at random in $\mathbb F_q$. 
\begin{itemize}
\item If $R_2<R_1D(q)/\log q$, it holds a.a.s. that
\begin{align*}
|\mathcal P_2(\ve s,\ve p_1)|=1.
\end{align*}
\item If $R_1D(q)/\log q\leq R_2\leq  D(q)$, it holds a.a.s. that
\begin{align*}
|\mathcal P_2(\ve s, \ve p_1)|\doteq 2^{n(R_2-R_1D(q)/\log q)}.
\end{align*}
\item If $R_2\geq D(q)$, it holds a.a.s. that
\begin{align*}
|\mathcal P_2(\ve s, \ve p_1)|\doteq 2^{n(D(q)-R_1D(q)/\log q)}.
\end{align*}
\end{itemize}
\label{lemma:size_p2}
\end{lemma}

\begin{proof}
We first consider the case when $R_2<R_1D(q)/\log q$. Recall in Lemma \ref{lemma:information_sum} we show that  for an information-sum $\ve s\in\mathcal A^{(k_2)}_{[W]}$, there are $L$  pairs of $(\ve m,\ve n')$ satisfying $\sf s(\ve m,\ve n')=\ve s$. In Lemma \ref{lemma:size_p1} we show that in the case  $R_2<R_1D(q)/ \log q$, all $L$ pair will give different parity-sum $\sf p_1$ asymptotically almost surely.  In other words for one parity-sum ${\sf p}_1$, there is only one pair  $(\ve m,\ve n')$ which gives ${\sf p}_1(\ve m,\ve n')=\ve p_1$, consequently there can be only one possible parity-sum ${\sf p}_2$  which results from this pair $(\ve m, \ve n')$, namely $|\mathcal P_2(\ve s, \ve p_1)|=1$. This proves the first claim.

Now we consider the remaining two cases.  It is shown in Appendix \ref{appendix:parity_sum}, Lemma \ref{lemma:check_distr} that  each entry of the parity sum ${\sf p}_2$ is i.i.d. according to $P_W$ hence the probability that the parity-sum ${\sf p}_2$ being atypical is negligible. For a given typical vector $\ve p\in\mathcal A^{(n-k_1)}_{[W]}$,   we define the random variable $Z_2(\ve p)$ to be the number of different pairs $(\ve m,\ve n')$, which give the parity-sum $\sf p_2$ equal to $\ve p$. In other words, define the random set
\begin{align*}
\mathcal Z_2(\ve p):=\{(\ve m, \ve n'): \sf p_2(\ve m,\ve n')=\ve p\}
\end{align*}
where each entry of $\ve H, \ve Q$ is chosen uniformly at random from $\mathbb F_q$, the random variable $Z_2(\ve p)$ is defined as $Z_2(\ve p):=|\mathcal Z_2(\ve p)|$. Now we study $Z_2(\ve p)$ for all pairs $(\ve m,\ve n')$ which satisfy $\sf s(\ve m, \ve n')=\ve s$ and $\sf p_1(\ve m,\ve n')=\ve p_1$ for some $\ve s$ and $\ve p_1$. Recall that in the proof of Lemma \ref{lemma:size_p1}, we have shown in (\ref{eq:Z1_estimates}) that if $\sf s(\ve m,\ve n')=\ve s$ for some $\ve s$ and if $R_2>R_1D(q)/\log q$, then the number of pairs $(\ve m,\ve n')$ satisfying $\sf p_1(\ve m, \ve n')=\ve p_1$ for some $\ve p_1$ is bounded as 
\begin{subequations}
\begin{align}
L'\leq 2^{n(R_2-R_1D(q)/\log q+\epsilon_n)}+2^{\frac{n}{2}(R_2-R_1D(q)/\log q+\epsilon_n')}\\
L'\geq 2^{n(R_2-R_1D(q)/\log q-\epsilon_n)}-2^{\frac{n}{2}(R_2-R_1D(q)/\log q+\epsilon_n')}
\end{align}
\label{eq:L_p}
\end{subequations}
Since it holds that $2^{\frac{n}{2}(R_2-R_1D(q)/\log q+\epsilon_n')}\leq 2^{n(R_2-R_1D(q)/\log q+\epsilon_n')}$ and $2^{\frac{n}{2}(R_2-R_1D(q)/\log q+\epsilon_n')}\leq \frac{1}{2}\cdot 2^{n(R_2-R_1D(q)/\log q+\epsilon_n')}$  for large enough $n$, we can conclude that
\begin{align}
L'\doteq 2^{n(R_2-R_1D(q)/\log q)}
\label{eq:L_p_approx}
\end{align}

Also recall Lemma \ref{lemma:2values} that under the condition that $\sf s(\ve m, \ve n')=\ve s, {\sf p}_1(\ve m, \ve n')=\ve p_1$ for some $\ve s, \ve p_1$, the possible parity sum ${\sf p}_2$ are constrained and  we have
\begin{align}
\pp{Z_2(\ve p)>1|\sf s(\ve m, \ve n')=\ve s, \sf p_1(\ve m, \ve n')=\ve p_1}&=\pp{Z_2(\ve p)>1|F(\ve a)}
\label{eq:Zh>1}
\end{align}
for some $\ve a$ depending only on $\ve s, \ve p_1, \ve H$ and $\ve Q$.

 
For the case  $R_1D(q)/\log q\leq R_2\leq D(q)$,  in Appendix \ref{app:Z_2}, we show that for a typical sequence $\ve p\in\mathcal F(\ve a)$, the  expectation and  variance of $Z_2(\ve p)$ conditioned on the event $F(\ve a)$ have the form
\begin{align}
2^{n(R_2-D(q)-\epsilon_n)}\leq \E{Z_2(\ve p)|F(\ve a)}\leq 2^{n(R_2-D(q)+\epsilon_n)}
\label{eq:E_Z_F}
\end{align}
for some $\epsilon_n\rightarrow 0$. Markov inequality implies that
\begin{align}
\pp{Z_2(\ve p)>1)|F(\ve a)}&\leq \E{Z_2(\ve  p)|F(\ve a)}\\
&\leq  2^{n(R_2-D(q)+\epsilon_n)}
\end{align}
which can be arbitrarily small with sufficiently large $n$ provided that $R_2\leq D(q)-2\epsilon_n$ and $n\epsilon_n\rightarrow \infty$. As $Z_2(\ve p)$ denotes the number of pairs $(\ve m, \ve n')$ which give a parity-sum part $\sf p_2$ equal to some vector $\ve p$, this means a.a.s. any party sum $\sf p_2$ can be formed by at most one pair $(\ve m, \ve n')$. In other words, every pair gives a distinct $\ve p$ a.a.s. hence the size of $\mathcal P_2(\ve s,\ve p_1)$ equals the total number of pairs $L'$ in (\ref{eq:L_p_approx}). This proves the first claim by letting $\epsilon_n\rightarrow 0$.

We then show that for the case $ R_2\geq  D(q)$ and conditioned on the event $F(\ve a)$, the random variable $Z_2(\ve p)$ concentrates around $\E{Z_2(\ve p)|F(\ve a)}$ for some typical sequence $\ve p\in\mathcal F(\ve a)$. 
For some $\epsilon'_n>0$ depending on $n$, by conditional Chebyshev inequality (see \cite[Ch. 23.4]{fristedt_modern_1997} for example) we have
\begin{align}
\pp{|Z_2(\ve p)-\E{Z_2(\ve p)|F(\ve a)}\geq 2^{\frac{n}{2}(R_2-D(q)+\epsilon'_n)}|F(\ve a)}&\leq \frac{\var{Z_2(\ve p)|F(\ve a)}}{2^{ 2\cdot\frac{n}{2}(R_2-D(q)\epsilon'_n)}}\\
&<\frac{\E{Z_2(\ve p)|F(\ve a)}}{2^{n(R_2-D(q)+\epsilon'_n)}}\\
&\leq 2^{-n(\epsilon'_n-\epsilon_n)}
\end{align}
where we used the inequality $\var{Z_2(\ve p)|F(\ve a)}\leq \E{Z_2(\ve p)|F(\ve a)}$ proved in Appendix \ref{app:Z_2}. If we choose $\epsilon'_n>\epsilon_n$ and $n$ such that $n(\epsilon'_n-\epsilon_n)\rightarrow\infty$ and $\epsilon'_n\rightarrow 0$ (this is possible because $\epsilon_n\searrow 0$), then a.a.s. $Z_2(\ve p)$ satisfies
\begin{align}
\E{Z_2(\ve p)|F(\ve a)}-2^{\frac{n}{2}(R_2-D(q)+\epsilon'_n)}\leq Z_2(\ve p)\leq   \E{Z_2(\ve p)|F(\ve a)}+2^{\frac{n}{2}(R_2-D(q)+\epsilon'_n)}
\label{eq:Z2_estimates}
\end{align}
conditioned on the event $F(\ve a)$. Furthermore we have the following identity regarding the total number of pairs $(\ve m,\ve n')$ satisfying $\sf s(\ve m,\ve n')=\ve s$ and $\sf p_1(\ve m,\ve n')=\ve p_1$ for some $\ve s, \ve p_1$:
\begin{align}
\sum_{\ve p\in \mathcal P_2(\ve s,\ve p_1)}Z_2(\ve p)=L'
\label{eq:total_number_Lp}
\end{align}
Combining (\ref{eq:Z2_estimates}) and (\ref{eq:total_number_Lp}),  the following estimates hold a.a.s.
\begin{align}
\frac{L'}{\E{Z_2(\ve p)|F(\ve a)}+2^{\frac{n}{2}(R_2-D(q)+\epsilon_n')}}\leq |\mathcal P_2(\ve s,\ve p_1)|\leq \frac{L'}{\E{Z_2(\ve p)|F(\ve a)}-2^{\frac{n}{2}(R_2-D(q)+\epsilon_n')}}
\end{align}
Using $L'$ from (\ref{eq:L_p}), Eq. (\ref{eq:E_Z_F}) and the above expression, $\mathcal P_2(\ve s,\ve p_1)$ is bounded a.a.s. as
\begin{align}
|\mathcal P_2(\ve s,\ve p_1)|&\geq \frac{2^{n(D(q)-R_1D(q)/\log q+o(1)))}-2^{\frac{n}{2}(-R_2+2D(q)-R_1D(q)/\log q+o(1))}}{1+2^{-\frac{n}{2}(R_2-D(q)+2\epsilon_n-\epsilon_n')}}\\
|\mathcal P_2(\ve s,\ve p_1)|&\leq \frac{2^{n(D(q)-R_1D(q)/\log q+o(1)))}+2^{\frac{n}{2}(-R_2+2D(q)-R_1D(q)/\log q+o(1))}}{1-2^{-\frac{n}{2}(R_2-D(q)-2\epsilon_n-\epsilon_n')}}
\end{align}
By the assumption that $R_2\geq D(q)$, we can let  $R_2=D(q)+\sigma_n$ for some $\sigma_n\rightarrow 0$, the two terms in the denumerators are
\begin{align}
2^{-\frac{n}{2}(R_2-D(q)+2\epsilon_n-\epsilon_n')}&=2^{-\frac{n}{2}(\sigma_n+2\epsilon_n-\epsilon')}\\
 2^{-\frac{n}{2}(R_2-D(q)-2\epsilon_n-\epsilon_n')}&=2^{-\frac{n}{2}(\sigma_n-2\epsilon_n-\epsilon_n')}
\end{align}
and both terms approaches $0$ if $\sigma_n>2\epsilon_n+\epsilon_n'$. Since both $\epsilon_n$ and $\epsilon_n'$ are chosen to approach $0$, we can let $\sigma_n$ approach  $0$ as well. Furthermore we have 
\begin{align*}
 2^{\frac{n}{2}(-R_2+2D(q)-R_1D(q)/\log q+o(1))}=2^{\frac{n}{2}(D(q)-R_1D(q)/\log q +\sigma_n+o(1))}
\end{align*}
We can conclude that for $R_2\geq D(q)$ and $n$ large enough we have a.a.s.
\begin{align*}
|\mathcal P_2(\ve s,\ve p_1)|&\leq \frac{2^{n(D(q)-R_1D(q)/\log q+o(1)))}+2^{\frac{n}{2}(D(q)-R_1D(q)/\log q +\sigma_n+o(1))}}{1-o_n(1)}\\
&\leq \frac{2\cdot 2^{n(D(q)-R_1D(q)/\log q+o(1))}}{1-o_n(1)}
\end{align*}
and
\begin{align*}
|\mathcal P_2(\ve s,\ve p_1)|&\geq \frac{2^{n(D(q)-R_1D(q)/\log q+o(1)))}-2^{\frac{n}{2}(D(q)-R_1D(q)/\log q +\sigma_n+o(1))}}{1+o_n(1)}\\
&\geq \frac{\frac{1}{2}\cdot 2^{n(D(q)-R_1D(q)/\log q+o(1)))}}{1+o_n(1)}
\end{align*}
since we have
\begin{align*}
2^{\frac{n}{2}(D(q)-R_1D(q)/\log q +\sigma_n+o(1))}\leq \frac{1}{2}\cdot 2^{n(D(q)-R_1D(q)/\log q +o(1))}
\end{align*}
for $n$ large enough.
Hence we can conclude that 
\begin{align} 
\frac{2^{n(D(q)-R_1D(q)/\log q+o(1))}}{1+o_n(1)}\leq |\mathcal P_2(\ve s,\ve p_1)|\leq \frac{2^{n(D(q)-R_1D(q)/\log q +o(1))}}{1-o_n(1)}
\end{align}
or equivalently $\mathcal P_2(\ve s,\ve p_1)\doteq 2^{n(D(q)-R_1D(q)/\log q)}$ a.a.s. if $n$ is sufficiently large. 
\end{proof}

Use the previous lemmas we can give the estimates on the size of the parity-sums $\mathcal P(\ve s)$.
\begin{lemma}[Estimates of $|\mathcal P(\ve s)|$]
Let $\ve m, \ve n'$ be two independent random vectors which are uniformly distributed in $\mathbb F_q^{k_1}$ and  $\mathbb F_q^{k_2}$, respectively.  For the pairs $(\ve m,\ve n')$  satisfying ${\sf s}(\ve m,\ve n')=\ve s$ for some $\ve s\in\mathcal A_{[W]}^{(k_2)}$,  let $\mathcal P(\ve s)$ denote the random set formed in (\ref{def:p}), where each entry of $\ve H$ and $\ve Q$ is chosen i.i.d. uniformly at random in $\mathbb F_q$. Then asymptotically almost surely it holds that
\begin{align*}
|\mathcal P(\ve s)|\doteq 2^{n(R_1+R_2-R_2H(W)/\log q)}
\end{align*}
for $R_2\leq D(q)$
and
\begin{align*}
|\mathcal P(\ve s)|\doteq 2^{n(R_1+D(q)-R_2H(W)/\log q)}
\end{align*}
for $R_2\geq D(q)$
\label{lemma:p_size}
\end{lemma}

\begin{proof}
When matrices $\ve H, \ve Q$ are generated randomly in the code construction, the relationship in (\ref{eq:relation_paritysum}) implies
\begin{align*}
|\mathcal P(\ve s)|=\sum_{\ve p_1\in \mathcal P_1(\ve s)}|\mathcal P_2(\ve s, \ve p_1)|
\end{align*}
where the cardinality of sets are random variables.

We first consider the case when $R_2<R_1D(q)/\log q$.  Lemma \ref{lemma:size_p1} shows that  for all $L$ pairs of $(\ve m,\ve n')$ satisfying ${\sf s}(\ve m,\ve n')=\ve s$, each of them gives a different parity-sum ${\sf p}_1$.  In Lemma \ref{lemma:size_p2} we also showed that $|\mathcal P_2(\ve s,\ve p_1)|=1$ in this case, hence a.a.s. we have
\begin{align*}
|\mathcal P(\ve s)|=|\mathcal P_1(\ve s)|\doteq 2^{n(R_1+R_2-R_2H(W)/\log q)}.
\end{align*}
where we use the result on $|\mathcal P_1(\ve s)|$ in Lemma \ref{lemma:size_p1}.

For the case when $R_2\geq R_1D(q)/\log q$,  Lemma \ref{lemma:size_p1} shows that  among $L$ pairs of $(\ve m,\ve n')$ satisfying ${\sf s}(\ve m,\ve n')=\ve s$, some of them give the same parity-sum ${\sf p}_1$. Using Lemma \ref{lemma:size_p1} and Lemma \ref{lemma:size_p2} we conclude that if $R_1D(q)/\log q \leq R_2< D(q)$,  we have a.a.s.
\begin{align*}
|\mathcal P(\ve s)|&=\sum_{\ve p_1\in\mathcal P_1(\ve s)} |\mathcal P_2(\ve s, \ve p_1)|\\
&\doteq \sum_{\ve p_1\in\mathcal P_1(\ve s)}  2^{n(R_2-R_1D(q)/\log q)}\\
&\doteq 2^{n(R_1-R_2)H(W)/\log q}\cdot 2^{n(R_2-R_1D(q)/\log q)}\\
&=2^{n(R_1+R_2-R_2H(W)/\log q)}
\end{align*}
and if $R_2\geq D(q)$, we have a.a.s.
\begin{align*}
|\mathcal P(\ve s)|&=\sum_{\ve p_1\in\mathcal P_1(\ve s)} |\mathcal P_2(\ve s, \ve p_1)|\\
&\doteq \sum_{\ve p_1\in\mathcal P_1(\ve s)}  2^{n(D(q)-R_1D(q)/\log q)}\\
&\doteq 2^{n(R_1-R_2)H(W)/\log q}\cdot 2^{n(D(q)-R_1D(q)/\log q)}\\
&=2^{n(R_1+D(q)-R_2H(W)/\log q)}
\end{align*}
which proves the claim.
\end{proof}

With the foregoing lemmas we can finalize the proof of Theorem \ref{thm:size_sumset_systematic}.
\begin{proof}[Proof of Theorem \ref{thm:size_sumset_systematic}]
We have assumed $R_1\geq R_2$ in all preceding proofs. Notice that the asymptotic estimates on $\mathcal P(\ve s)$ in Lemma \ref{lemma:p_size} hold for \textit{all} typical information-sum $\ve s$ in (\ref{eq:K_normal}). 
Hence combining Lemma \ref{lemma:typical_matter} and Lemma \ref{lemma:p_size}, we conclude that  for $R_2\leq D(q)$ we have a.a.s.
\begin{align}
|\mathcal K_N|&=\sum_{\ve s\in \mathcal A^{(k_2)}_{[W]}}|\mathcal P(\ve s)|\\
&\doteq  2^{k_2H(W)}\cdot 2^{n(R_1+R_2-R_2H(W)/\log q)}\\
&=2^{n(R_1+R_2)}.
\end{align}
where we have used the fact that $\left|\mathcal A^{(k_2)}_{[W]}\right|=2^{k_2(H(W)+o(1))}$ from  Lemma \ref{lemma:typical_sequence}. For $R_2\geq D(q)$ we have a.a.s.
\begin{align}
|\mathcal K_N|&\doteq 2^{k_2H(W)}\cdot 2^{n(R_1+D(q)-R_2H(W)/\log q)}\\
&=2^{n(R_1+D(q))}.
\end{align}
In the case when $R_1\leq R_2$, similar results is obtained by simply switching $R_1, R_2$. Namely in this case we have
\begin{align}
|\mathcal K_N|&\doteq
\begin{cases}
2^{n(R_1+R_2)} &\text{if } R_1\leq D(q)\\
2^{n(R_2+D(q))} &\text{if }R_1> D(q)
\end{cases}
\end{align}
Lastly it can be verified straightforwardly that for any $R_1,R_2\in[0,\log q]$ we can combine the expressions above into one compact formulation as
\begin{align}
|\mathcal K_N| \doteq \min\left\{2^{n(R_1+R_2)}, 2^{n\left(\max\{R_1,R_2\}+D(q)\right)}\right\}
\end{align}

Now we prove the asymptotic equipartion property (AEP) of the normal typical sumset $\mathcal K_N$ in (\ref{eq:AEP}).  Let $M^{k_1}$ denote a $k_1$-length random vector uniformly distributed in $\mathcal U^{k_1}$ and $N^{k_2}$ a  $k_2$-length random vector uniformly distributed on $\mathcal U^{k_2}$.  If we view $M^{k_1},N^{k_2}$ as two independent messages and let $T_1^n, T_2^n$ be two codewords generated using $M^{k_1}, N^{k_2}$,  then $T_1^n, T_2^n$ are two independent random variables uniformly distributed on $\mathcal C'_1:=\mathcal C_1\oplus \ve d_1,\mathcal C_2$ respectively. We assume that with the chosen $\ve Q, \ve H$, $\mathcal C_1'\oplus \mathcal C_2$ has a normal typical sumsets $\mathcal K_N$. Recall that  $P_{\mathcal S}$ denotes the probability distribution on the sumset $\mathcal C_1'+\mathcal C_2$ induced by $T_1^n, T_2^n$ as in Definition \ref{def:P_S}.

Again assume $R_1\geq R_2$, we first consider the rate regime when $R_2\leq D(q)$.  In this case Lemma \ref{lemma:p_size} shows that the number of possible parity-sum $\sf p$ is equal to the number of pairs $(\ve m,\ve n')$ satisfying $\sf s(\ve m,\ve n')=\ve s$. In other words any sum codewords $\ve w\in\mathcal K_N$ is formed by a unique pair, say, $(\ve m_0, \ve n_0)$. Hence
\begin{align}
P_{\mathcal S}(\ve w)&=\pp{M^{k_1}=\ve m_0, N^{k_2}=\ve n_0}\\
&=\pp{M^{k_1}=\ve m_0}\pp{ N^{k_2}=\ve n_0}\\
&=q^{-k_1}\cdot q^{-k_2}=2^{-n(R_1+R_2)}
\end{align}

Now consider the case when $R_2\geq   D(q)$. Lemma \ref{lemma:p_size} shows that among all $L$ pairs of $(\ve m,\ve n')$ satisfying $\sf s(\ve m,\ve n')=\ve s$ for some $\ve s$, many pairs give the same parity-sum $\ve p$. More precisely, let $Z(\ve s, \ve p)$ denote the number of pairs $(\ve m,\ve n')$ sum up to a particular parity-sum $\ve p:=(^{\ve p_1}_{\ve p_2})$ given $\sf s(\ve m,\ve n')=\ve s$. We have shown in (\ref{eq:Z2_estimates}) that given the constraints that $\sf s(\ve m,\ve n')=\ve s$ and ${\sf p}_1(\ve m,\ve n')=\ve p_1$, then the number of $(\ve m,\ve n')$ satisfying ${\sf p}_2(\ve m,\ve n')=\ve p_2$ for some $\ve p_2$ is bounded  as
\begin{align*}
Z_2(\ve p)&\geq 2^{n(R_2-D(q)-\epsilon_n)}-2^{\frac{n}{2}(R_2-D(q)+\epsilon'_n)}\\
 Z_2(\ve p)&\leq   2^{n(R_2-D(q)+\epsilon_n)}+2^{\frac{n}{2}(R_2-D(q)+\epsilon'_n)}
\end{align*}
Notice this is also the number of pairs $(\ve m,\ve n')$ sum up to a particular parity-sum $\ve p:=(^{\ve p_1}_{\ve p_2})$ given $\sf s(\ve m,\ve n')=\ve s$. Hence we have
\begin{align*}
Z(\ve s, \ve p)\doteq 2^{n(R_2-D(q))}
\end{align*}
Hence for a sum codeword $\ve w=(^{ \ve s}_{\ve p})\in \mathcal K_N$, we have
\begin{align}
P_{\mathcal S}(\ve w)&=\sum_{\stackrel{(\ve m, \ve n)}{\sf s(\ve m,\ve n')=\ve s, \sf p(\ve m,\ve n')=\ve p}}\pp{M^{k_1}=\ve m, N^{k_1}=\ve n}\\
&=\sum_{\stackrel{(\ve m, \ve n)}{\sf s(\ve m,\ve n')=\ve s, \sf p(\ve m,\ve n')=\ve p}}q^{-(k_1+k_2)}\\
&\doteq 2^{n(R_2-D(q))}\cdot 2^{-n(R_1+R_2)}\\
&=2^{-n(R_1+D(q))}
\end{align}
The exact arguments hold for the case when $R_1\leq R_2$, and this concludes the proof of the AEP  and  Theorem \ref{thm:size_sumset_systematic}.
\end{proof}

With the results established for systematic linear codes,  we can finally prove the  results for general linear codes. 

\begin{proof}[Proof of Theorem \ref{thm:size_sumset}]
Assume $k_1\geq k_2$. We first fix $\mathcal C_1$ and consider the construction of $\mathcal C_2$. In the construction (\ref{eq:nestedcodes}), the code ensemble $\mathcal C_2$ is constructed using $k_2$ linearly independent basis of $\mathcal C_1$. In (\ref{eq:nestedcodes}) we used the first $k_2$ columns of $\ve G$, however since each entry of $\ve G$ is chosen i.i.d. uniformly,  by symmetry we will have the same ensemble if we choose any $k_2$ linearly independent basis of the code $\mathcal C_1$.  Now consider the construction of $\mathcal C_2$ in (\ref{eq:systematic_proof}). In Theorem \ref{thm:size_sumset_systematic} we considered the ensemble of codes generated as in (\ref{eq:systematic_proof}) where each entry of $\ve Q$ and $\ve H$ is chosen i.i.d.   according to the uniform distribution in $\mathbb F_q$.  We first show that the ensemble of $\mathcal C_2$ generated in (\ref{eq:systematic_proof}) can be equivalently rewritten in the following way
\begin{align}
\mathcal C_2=\left\{ \ve v: \ve v=\begin{bmatrix}
\ve I_{k_2\times k_2}\\
\ve Q'
\end{bmatrix}
\ve n', \mbox{ for all }\ve n'\in\mathbb F_q^{k_2}\right\}
\label{eq:C2_equivalent}
\end{align}
for some $\ve Q'\in\mathbb F_q^{(n-k_2)\times k_2}$. To see what is the matrix $\ve Q'$, using $\ve g_i$ to denote the $i$-th column of the matrix $\begin{bmatrix}
\ve I_{k_1\times k_1}\\
\ve Q
\end{bmatrix}$, using $\ve H_{i,j}$ to denote the $(i,j)$ entry of $\ve H$ and $\ve n'_i$ the $i$-th entry of $\ve n'$,  we can rewrite $\ve v$ in (\ref{eq:systematic_proof}) as 
\begin{align*}
\ve v&=\ve g_1\ve n'_1\oplus\ldots \oplus \ve g_{k_2}\ve n_{k_2}'\oplus \ve g_{k_2+1}\sum_{i=1}^{k_2}\ve H_{1i}\ve n'_i\oplus \ldots \oplus \ve g_{k_1}\sum_{i=1}^{k_2}\ve H_{k_1-k_2,i} \ve n'_i\\
&=\left(\ve g_1\oplus \sum_{i=1}^{k_1-k_2}\ve H_{i1}\ve g_{k_2+i}\right)\ve n'_1\oplus \left(\ve g_2\oplus \sum_{i=1}^{k_1-k_2}\ve H_{i2}\ve g_{k_2+i}\right)\ve n'_2\oplus \ldots \oplus \left(\ve g_{k_2}\oplus \sum_{i=1}^{k_1-k_2}\ve H_{ik_2}\ve g_{k_2+i}\right)\ve n'_{k_2}
\end{align*}
This shows that the $j$-th column of $\begin{bmatrix}
\ve I_{k_2\times k_2}\\
\ve Q'
\end{bmatrix}$
is given by  $\ve g_j\oplus \sum_{i=1}^{k_1-k_2}\ve H_{ij}\ve g_{k_2+i}$. Notice that $\ve g_{i}, i=1,\ldots, k_1$ is a vector whose first $k_1$ entries are all zero except that its $i$-th position is $1$, and the remaining entries are chosen i.i.d. uniformly from $\mathbb F_q$. Then the vector $\sum_{i=1}^{k_1-k_2}\ve H_{ij}\ve g_{k_2+i}$, for $j=1,\ldots, k_2$ has zero entries for the first $k_2$ positions. Hence indeed $\ve g_j\oplus \sum_{i=1}^{k_1-k_2}\ve H_{ij}\ve g_{k_2+i}$ has zero entries for the first $k_2$ entries except that it has $1$ at the $j$-th position, and the last $k_1-k_2$ entries  given by $\ve H$ and $\ve g_i,i=1,\ldots, k_1-k_2$. This proves that $\mathcal C_2$ can be generated equivalently as in (\ref{eq:C2_equivalent}). Furthermore, since each entry of $\ve H$ is chosen i.i.d. uniformly, it follows that each entry of $\ve Q'$ is also chosen i.i.d. according to the uniform distribution in $\mathbb F_q$. This also shows that in the construction (\ref{eq:systematic_proof}), $\mathcal C_2$ is generated using $k_2$ linearly independent basis of $\mathcal C_1$.

It is known that the systematic generator matrix for a systematic linear code is unique.  Furthermore, as we can identify an $(n,k_1)$-linear code with the $k_1$-dimensional subspace spanned by its generator matrix, each systematic generator matrix thus gives a unique $k_1$-dimensional subspace. It is also known that the total number of $k_1$-dimensional subspaces in $\mathbb F_q^n$ is given by the so-called Gaussian binomial coefficient (see \cite{roman_advanced_2005} for example):
\begin{align}
{n\choose k_1}_q:=\frac{(q^n-1)(q^n-q)\cdots (q^n-q^{k_1-1})}{(q^k_1-1)(q^k_1-q)\cdots (q^{k_1}-q^{k_1-1})}
\end{align}

Let  $\mathcal C'_j, j=1,2 $ be the corresponding systematic linear code of an  arbitrary $(n, k_j)$-linear code by permuting the entries and assume that $\mathcal C_2\subseteq\mathcal C_1$.   Lemma \ref{lemma:equivalence} shows that there is a one-to-one mapping between $\mathcal C_1+\mathcal C_2$ and $\mathcal C_1'+\mathcal C_2'$. Hence if codes $\mathcal C_1, \mathcal C_2$ are equivalent to some systematic linear codes $\mathcal C'_1, \mathcal C_2'$ with a normal typical sumset $\mathcal K_N$, the codes $\mathcal C_1, \mathcal C_2$ also have a normal typical sumset. By identifying a codebook with its corresponding subspace,  Theorem \ref{thm:size_sumset_systematic}  shows that almost all of the $k_1$-dimensional subspaces (with a $k_2$-dimensional subspace within it generated by choosing any $k_2$ linearly independent basis) have a normal typical sumset, since every linear code is equivalent to some systematic linear code.  Formally the number of codes $\mathcal C_1$ (with $\mathcal C_2$ generated with $k_2$ linearly independent basis of $\mathcal C_1$) which have a normal typical sumset  is $(1-o(1)){n\choose k_1}_q$.


Now consider the codes ensemble in Theorem \ref{thm:size_sumset} where we choose all possible $q^{nk_1}$ generator matrices  with equal probability. Clearly some of the generator matrices give the same code if they span the same $k_1$-dimensional subspace. We now show most of these generator matrices will give codes which have a normal typical sumsets. Notice that each distinct $k_1$-dimensional subspace can be generated by $(q^{k_1}-1)(q^{k_1}-q)\cdots(q^{k_1}-q^{k_1-1})$ different generator matrices (because there are this many different choices of basis in a $k_1$-dimensional subspace).  Hence the fraction of the generator matrices with a normal typical sumset is 
\begin{align*}
\rho:=\frac{(1-o(1)){n\choose k_1}_q\cdot(q^{k_1}-1)(q^{k_1}-q)\cdots(q^{k_1}-q^{k_1-1})}{q^{nk_1}}&=(1-o(1))\frac{(q^n-1)(q^n-q)\cdots (q^n-q^{k_1-1})}{q^{nk_1}}\\
&=(1-o(1))(1-q^{-n})(1-q^{-n+1})\cdots(1-q^{-n+k_1-1})\\
&>(1-o(1))(1-q^{-n+k_1})^{k_1}
\end{align*}
Assume $k_1=\beta n$ for some $\beta\in[0,1)$,  L'H\^{o}pital's rule shows the logarithm of the term $(1-q^{-n+k_1})^{k_1}$ has limit 
\begin{align}
\lim_{n\rightarrow\infty}\beta n\ln (1-q^{-n(1+\beta)})&=\lim_{n\rightarrow\infty} \frac{\ln (1-q^{-n(1+\beta)})}{1/\beta n}\\
&=\lim_{n\rightarrow\infty} \frac{-\beta n^2}{1-q^{-n(1+\beta)}}q^{-n(1+\beta)}(1+\beta)\ln q\\
&=0
\end{align}
Hence the  fraction $\rho$  of codes with a normal typical sumset is arbitrarily close to $1$ for sufficiently large $n$. This proves that the code ensemble considered in Theorem \ref{thm:size_sumset} have a normal typical sumset a.a.s..

The proof of AEP property of the normal typical sumset is the same as in the proof of Theorem \ref{thm:size_sumset_systematic} by using the fact that every linear code is equivalent to some systematic linear code,  and we shall not repeat it.
\end{proof}


\section{Application to computation over multiple access channels}
\label{sec:computation}

In this section we  study a computation problem over noisy multiple access channels. We consider a general two-user discrete memoryless multiple access channel described by a conditional probability distribution $P_{Y|X_1X_2}$ with input and output alphabets  $\mathcal X_1,\mathcal X_2$ and $\mathcal Y$, respectively. Unlike the usual coding schemes, we always assume that codebooks $\mathcal C_1, \mathcal C_2$ are subsets of $\mathbb F_q^n$ (or $\mathcal U^n$), such that the (entry-wise) addition of codewords is well-defined. 

A $(2^{nR_1}, 2^{nR_2}, n)$ computation code in $\mathcal U^n$ for a two-user MAC consists of
\begin{itemize}
\item two message sets $[1:2^{nR_1}]$ and $[1:2^{nR_2}]$,
\item two encoders, where encoder $1$ first assigns a codeword $\ve t(\ve m)\in\mathcal U^n$ to each message $\ve m\in[1:2^{nR_1}]$ and then map the codeword $\ve t$ to a channel input $\ve x\in\mathcal X_1^n$. The operation of encoder $2$  is the same.
\item a decoder $\mathcal D$ which assigns an estimated sum of codewords $\hat{\ve w}\in\mathcal W^n$ for each channel output $\ve y\in\mathcal Y^n$.
\end{itemize}

We assume that the messages $M, N$ from two users are uniformly chosen from the message sets. The \textit{average sum-decoding error probability} as
\begin{align}
P_e^{(n)}:=\sum_{\ve m,\ve n}\pp{M=\ve m,N=\ve n}\lambda(\ve m,\ve n)
\label{eq:error_proba}
\end{align}
where $\lambda(\ve m,\ve n)$ to denote the conditional sum-decoding error probability of this code if $\ve t(\ve m)+\ve v(\ve n)$ is the true sum codeword, i.e.
\begin{align}
\lambda(\ve m,\ve n):=\pp{\mathcal D(Y^n)\neq \ve t(\ve m)+\ve v(\ve n)| M=\ve m, N=\ve n}
\end{align}
A \textit{computation rate pair} $(R_1,R_2)$ is said to be achievable if there exists a sequence of $(2^{nR_1}, 2^{nR_2},n)$ computation codes in $\mathcal U^n$ such that $\lim_{n\rightarrow \infty}P_e^{(n)}=0$.

Similar problem has been studied using the compute-and-forward scheme \cite{NazerGastpar_2011} and  nested linear codes (\cite{padakandla_computing_2013}\cite{Nazer_compute_discrete_2014}\cite{ZhuGastpar_ITW}) where the modulo sum $\ve t\oplus \ve v$ is to be decoded. Here we study the problem of decoding the integer sum $\ve t+\ve v$ directly. First notice that the integer sum $\ve t+\ve v$ always allow us to recover the modulo sum $\ve t\oplus \ve v$. Another reason for insisting on decoding the integer sum is that it could be  more useful than a modulo sum in some scenario.  For example, consider an additive  interference network with multiple transmitter-receiver pairs where all transmitted signals are added up at receivers.  Because of the additivity of the channel, each receiver experiences interference which is the sum of signals of all other transmitters. In this case it is of interest to be able to decode the sum of the codewords because this is exactly the total interference each receiver suffers.


\begin{theorem}[Achievable computation rate pairs]
A computation rate pair $(R_1,R_2)$  is achievable in the two-user multiple access channel if it satisfies
\begin{align}
\max\{R_1,R_2\}&< I(U_1+U_2;Y)-D(q)
\label{eq:achievable_rate}
\end{align}
where $U_1, U_2$ are independent random variables with distribution  $P_U$ defined in (\ref{eq:P_U}) and the joint distribution $P_{U_1U_2Y}$ is given by $P_{U_1U_2Y}(u_1,u_2,y)=\sum_{x_1,x_2}P_U(u_1)P_U(u_2)P_{X_1|U}(x_1|u_1)P_{X_2|U}(x_2|u_2)P_{Y|X_1X_2}(y|x_1,x_2)$. Let $P_{X_1|U}, P_{X_2|U}$ be two arbitrary conditional probability distribution functions where $U$ and $X_1$ (resp. $X_2$) take values in $\mathcal U$ and $\mathcal X_1$ (resp. $\mathcal X_2$). The function $D(q)$ is defined in (\ref{eq:Dq}).
\label{thm:computation_rate}
\end{theorem}


\begin{proof}
We provide the details of the proof by starting with the coding scheme:
\begin{itemize}
\item \textbf{Codebook generation.} Let $k_j=\lfloor nR_j/\log q\rfloor, j=1,2$ and represent messages from user $j$ using all $k_j$-length vectors in $\mathcal U^{k_j}$.  Assume $k_1\geq k_2$, for messages $\ve m$ from user 1 and $\tilde{\ve n}$ from user 2 we generate nested linear codes as
\begin{subequations}
\begin{align}
\mathcal C_1&:=\left\{\ve t :\ve t=\ve G\ve m\oplus \ve d_1, \mbox{ for all }\ve m\in \mathbb F_q^{k_1}\right\}\\
\mathcal C_2&:=\left\{\ve v :\ve v=\ve G\ve n=\ve G\begin{bmatrix}
\tilde{\ve n}\\ 
\ve 0
\end{bmatrix} \oplus \ve d_2, \mbox{ for all }\tilde{\ve n}\in \mathbb F_q^{k_2} \right\}
\end{align}
\label{eq:codebook_codingtheorem}
\end{subequations}
for some generator matrix $\ve G$ and two $n$-length vectors $\ve d_1, \ve d_2$. We use $\mathcal K_N$ to denote a normal typical sumset of $\mathcal C_1+\mathcal C_2$, if it exists.
\item \textbf{Encoding.} Fix two arbitrary conditional probability distribution functions $P_{X_1|U}, P_{X_2|U}$ where $U$ takes values in $\mathcal U$ and $X_1, X_2$ takes value in $\mathcal X_1,\mathcal X_2$, respectively.  Given a chosen message $\ve m$, user 1 picks the corresponding codeword $\ve t(\ve m)$ generated above, and transmit $\ve x_{1,i}(\ve t_i)$ at time $i$ where $\ve x_{1,i}$ is generated according to $P_{X_1|U}(\ve x_{1,i}|\ve t_i)$ independently for all $i=1,\ldots, n$. User 2 carries out the same encoding steps.
\item \textbf{Decoding.} Upon receiving the channel output $\ve y$, the decoder declares the sum codeword to be $\hat{\ve w}$ if it can find a unique $\hat{\ve w}$ satisfying the following
\begin{align}
(\hat{\ve w}, \ve y)\in \mathcal A_{[WY]}^{(n)} \mbox{ with } \hat{\ve w}\in\mathcal K_N\backslash\mathcal L
\label{eq:typicalitycheck_sumdecoding}
\end{align}
where the joint distribution $P_{WY}$  is defined as $P_{WY}(w,y):=\sum_{u_1,u_2}P_{U_1U_2Y}(u_1,u_2,y)\cdot \ve 1_{w=u_1+u_2}$ and the set $\mathcal L$ is defined as 
\begin{align}
\mathcal L=\{\ve G\ve m\oplus \ve d_1+\ve G\ve n\oplus \ve d_2, \ve m=c\ve n \mbox{ for some }c\in\mathbb F_q\}
\end{align}
Namely $\mathcal L$ contains the sum codewords resulting from two messages $\ve m, \ve n$ which are linearly dependent.   Otherwise an error is declared for the decoding process.

\end{itemize}

\textbf{Analysis of the probability of error.}  We analyze the average error probability over an ensemble of codes, namely the ensemble where the each entry of the generator matrix $\ve G$ and \textit{dither vectors} $\ve d_1, \ve d_2$  are generated independently and uniformly from $\mathbb F_q$. First notice that we can assume that two linearly independent messages $\ve m_1,\ve n_1$ are chosen and the corresponding channel inputs are used. To see this,  we rewrite the average sum-decoding error probability in (\ref{eq:error_proba}) for some $c\in\mathbb F_q$ as
\begin{align*}
P_e^{(n)}&\leq \sum_{(\ve m, \ve n): \ve m\neq c\ve n}\pp{M=\ve m, N=\ve n}\lambda(\ve m, \ve n)+\pp{M=c\cdot N}\\
&\leq \sum_{(\ve m, \ve n): \ve m\neq c\ve n}\pp{M=\ve m, N=\ve n}\lambda(\ve m, \ve n)+q\cdot 2^{-n\min(R_1,R_2)}
\end{align*}
and the last term vanish for positive rates $R_1,R_2$ and large enough $n$.


%

In the following we use $W^n(\ve m,\ve n)$ to denote $G\ve m\oplus d_1+G\ve n\oplus d_2$ with  randomly chosen $G, d_1, d_2$ and the true sum is $W_1^n:=G\ve m_1\oplus d_1+G\ve n_1\oplus d_2$ where the chosen message $\ve m_1, \ve n_1$ are linearly independent.  When consider the conditional error probability $\lambda(\ve m_1,\ve n_1)$,  there are three kinds of errors:
\begin{align*}
\mathcal E_1&:=\{\text{the codes generated by }G, d_1, d_2 \text{ does not have a normal typical sumset } \mathcal K_N\}\\
\mathcal E_2&:=\{W_1\notin \mathcal K_N\backslash\mathcal L\}\cap \overline{\mathcal E_1}\\
\mathcal E_3&:=\{(\tilde W^n,Y^n)\in\mathcal A_{[WY]}^{(n)}\text{ for some } \tilde W^n\in\mathcal K_N\backslash \mathcal L, \tilde W^n\neq W_1^n  \}\cap \overline{\mathcal E_1}.
\end{align*}
To lighten the notation we define the event $\mathcal M:=\{M=\ve m_1,N=\ve n_1\}$. Using the union bound we can upper bound the conditional sum-decoding error probability as
\begin{align}
\lambda(\ve m_1,\ve n_1)\leq \pp{\mathcal E_1|\mathcal M}+\pp{\mathcal E_2|\mathcal M}+\pp{\mathcal E_3|\mathcal M}
\label{eq:conditional_error}
\end{align}

It holds for  the error event $\mathcal E_1$ that
\begin{align}
\pp{\mathcal E_1|\mathcal M}=\pp{\mathcal E_1}\leq o(1).
\label{eq:error_1}
\end{align}
Indeed, it is easy to see that the sumset of $\mathcal C_1, \mathcal C_2$ generated in (\ref{eq:codebook_codingtheorem}) has the same size as the sumset of $\mathcal C_1\oplus \ve d_1, \mathcal C_2\oplus \ve d_1$. But Theorem \ref{thm:size_sumset} shows that a.a.s., the codes $\mathcal C_1\oplus \ve d_1, \mathcal C_2\oplus \ve d_1$ generated by a randomly chosen $G$ and any $d_1, d_2$ has a normal typical sumset $\mathcal K_N$.  It also holds for the error event $\mathcal E_2$ that 
\begin{align}
\pp{\mathcal E_2|\mathcal M}= \pp{W_1^n(\ve m_1,\ve n_1)\notin \mathcal K_N\backslash \mathcal L|\mathcal M}\leq o(1)
\label{eq:error_2}
\end{align}
because by the definition of the typical sumset,  the true sum codeword $W_1^n=G\ve m_1+G\ve n_1$ should fall into $\mathcal K_N$ a.a.s.. Also by the assumption that $\ve m_1, \ve n_1$ are linearly independent, the true sum $W^n_1$ does not belong to $\mathcal L$.

To investigate the error event $\mathcal E_3$, we further divide the set of $\tilde W^n$ satisfying the condition in  event $\mathcal E_3$ into the following three subclasses:
\begin{align*}
\mathcal B_{1}&:=\{\tilde W^n(\ve m',\ve n')\in\mathcal K_N\backslash\mathcal L: (\tilde W^n(\ve m',\ve n'),Y^n)\in\mathcal A_{[WY]}^{(n)}, \tilde W^n\neq W_1^n, \ve m'\neq \ve m_1, \ve m'\neq \ve n_1, \ve n'\neq \ve n_1, \ve n'\neq \ve m_1 \}\\
\mathcal B_2&:=\{\tilde W^n(\ve m_1,\ve n')\in\mathcal K_N\backslash\mathcal L: (\tilde W^n(\ve m_1,\ve n'),Y^n)\in\mathcal A_{[WY]}^{(n)},\tilde W^n\neq W_1^n, \ve n'\neq \ve n_1 \}\\
\mathcal B_3&:=\{\tilde W^n(\ve m',\ve n_1)\in\mathcal K_N\backslash\mathcal L: (\tilde W^n(\ve m',\ve n_1),Y^n)\in\mathcal A_{[WY]}^{(n)},\tilde W^n\neq W_1^n,  \ve m'\neq \ve m_1 \}
\end{align*}
Based on the decoding rule (\ref{eq:typicalitycheck_sumdecoding}) and the above classification,  we  can express the last term in  (\ref{eq:conditional_error}) as
\begin{align}
\pp{\mathcal E_3|\mathcal M}&\leq \pp{(\tilde W^n,Y^n)\in\mathcal A_{[WY]}^{(n)}\text{ for some } \tilde W^n\in\mathcal K_N\backslash\mathcal L, \tilde W^n\neq W_1^n |\mathcal M}\\
&\leq \sum_{j=1}^3\pp{\bigcup_{\tilde W^n \in\mathcal B_j} (\tilde W^n, Y^n)\in\mathcal A_{[WY]}^{(n)}\mid \mathcal M}
\label{eq:error_3}
\end{align} 
and analyze each term separately. For all $j=1,2,3$,  we can rewrite the term  $\pp{\cup_{\tilde W^n\in\mathcal B_j}(\tilde W^n, Y^n)\in\mathcal A^{(n)}_{[WY]}|\mathcal M}$ in the following way
\begin{align}
&\pp{\bigcup_{\tilde W^n\in\mathcal B_j}(\tilde W^n, Y^n)\in\mathcal A_{[WY]}^{(n)}\mid\mathcal M}\leq \sum_{\tilde W^n\in\mathcal B_j}\pp{\tilde W^n,Y^n)\in \mathcal A_{[WY]}^{{(n)}}\middle| \mathcal M}\\
&= \sum_{\tilde W^n\in\mathcal B_j}  \pp{Y^n\in \mathcal A_{[Y]}^{(n)}|\mathcal M}\pp{\tilde W^n,Y^n)\in \mathcal A_{[W|Y]}^{(n)}(\ve y)|Y^n\in \mathcal A_{[WY]}^{(n)},\mathcal M}\\
&= \sum_{\tilde W^n\in\mathcal B_j} \sum_{\ve y\in \mathcal A_{[Y]}^{(n)}}\pp{Y^n=\ve y|\mathcal M}\sum_{\ve w_0\in \mathcal A_{[W|Y]}^{(n)}(\ve y)}\pp{\tilde W^n=\ve w_0|Y^n=\ve y,\mathcal M}
\label{eq:common_error}
\end{align}
We show in Appendix \ref{sec:independent} that for $\tilde W^n(\ve m',\ve n')\in\mathcal B_1$ we have 
\begin{align}
\pp{\tilde W^n(\ve m',\ve n')=\ve w_0|Y^n=\ve y,\mathcal M}=\pp{\tilde W^n(\ve m',\ve n')=\ve w_0|\mathcal M}
\label{eq:proof_independent}
\end{align}
Namely,  $\tilde W^n(\ve m',\ve n')\in\mathcal B_1$ are (conditionally) independent from $Y^n$. Hence we can continue (\ref{eq:common_error}) as
\begin{align}
\pp{\bigcup_{\tilde W^n\in\mathcal B_1}(\tilde W^n, Y^n)\in\mathcal A_{[WY]}^{(n)}\mid\mathcal M}&\leq\sum_{\tilde W^n(\ve m',\ve n')\in\mathcal B_1} \sum_{\ve y\in \mathcal A_{[Y]}^{(n)}}\pp{Y^n=\ve y|\mathcal M}\sum_{\ve w_0\in \mathcal A_{[W|Y]}^{(n)}(\ve y)}\pp{\tilde W^n(\ve m',\ve n')=\ve w_0|\mathcal M}\\
&=\sum_{\tilde W^n(\ve m',\ve n')\in\mathcal B_1}  \pp{Y^n\in \mathcal A_{[Y]}^{(n)}|\mathcal M}\pp{\tilde W^n(\ve m',\ve n'),Y^n)\in \mathcal A_{[W|Y]}^{(n)}(\ve y)|\mathcal M}\\
&=\sum_{\tilde W^n(\ve m',\ve n')\in\mathcal B_1}  \pp{(\tilde W^n(\ve m',\ve n'),Y^n)\in \mathcal A_{[W,Y]}^{(n)}|\mathcal M}\\
&\leq |\mathcal K_N|2^{-n(I(W;Y)-\epsilon_n)}
\label{eq:bound_B_1}
\end{align}
where the last inequality follows as for independent $\tilde W^n, Y^n$ we have (see e.g. \cite[Ch. 2.5]{Elgamal_Kim_2011})
\begin{align}
\pp{(\tilde W^n(\ve m',\ve n'),Y^n)\in \mathcal A_{[W,Y]}^{(n)}|\mathcal M}\leq 2^{-n(I(W;Y)-\epsilon_n)}
\end{align}
and the fact that $|\mathcal B_1|\leq |\mathcal K_N|$.

In Appendix \ref{sec:independent} we also show that for $\tilde W^n(\ve m_1,\ve n')\in\mathcal B_2$ we have
\begin{align}
\pp{\tilde W^n(\ve m_1,\ve n')=\ve w_0|Y^n=\ve y,\mathcal M}=q^{-n}
\end{align}
and the same for  $\tilde W^n(\ve m',\ve n_1)\in\mathcal B_3$.

To bound   $\pp{\cup_{\tilde W^n\in\mathcal B_2}(\tilde W^n, Y^n)\in\mathcal A|\mathcal M}$, we continue with (\ref{eq:common_error}) as
\begin{align}
&\pp{\bigcup_{\tilde W^n\in\mathcal B_2}(\tilde W^n, Y^n)\in\mathcal A_{[WY]}^{(n)}\mid\mathcal M}\\
&\leq\sum_{\tilde W^n(\ve m_1,\ve n')\in\mathcal B_2} \sum_{\ve y\in \mathcal A_{[Y]}^{(n)}}\pp{Y^n=\ve y|\mathcal M}\sum_{\ve w_0\in \mathcal A_{[W|Y]}^{(n)}(\ve y)}\pp{\tilde W^n(\ve m_1,\ve n')=\ve w_0|Y^n=\ve y,\mathcal M}\\
&=\sum_{\tilde W^n(\ve m_1,\ve n')\in\mathcal B_2}  \sum_{\ve y\in \mathcal A_{[Y]}^{(n)}}\pp{Y^n=\ve y|\mathcal M}\sum_{\ve w_0\in \mathcal A_{[W|Y]}^{(n)}(\ve y)}q^{-n}\\
&\leq \sum_{\tilde W^n(\ve m_1,\ve n')\in\mathcal B_2} \sum_{\ve w_0\in \mathcal A_{[W|Y]}^{(n)}(\ve y)}q^{-n}\\
&\leq 2^{nR_2}\cdot 2^{n(H(W|Y)+\delta_n)}2^{-n\log q}\\
&= 2^{-n(R_2+H(W|Y)-\log q+\delta_n)}
\label{eq:bound_B_2}
\end{align}
where we have used the fact the cardinality  of the conditional typical set $\mathcal A^{(n)}_{[W|Y]}(\ve y)$ is upper bounded by $2^{n(H(W|U)+\delta_n)}$ for some $\delta_n\rightarrow 0$ and the fact that the number of sums of the form $\tilde W^n(\ve m_1,\ve n')$ is upper bounded by $2^{nR_2}$ because $\ve n'$ can only take $2^{nR_2}$ many values.
Using a similar argument we can show that
\begin{align}
\pp{\bigcup_{\tilde W^n(\ve m',\ve n_1)\in\mathcal B_3}(\tilde W^n(\ve m',\ve n_1), Y^n)\in\mathcal A_{[W,Y]}^{(n)}|\mathcal M}\leq 2^{-n(R_1+H(W|Y)-\log q+\delta_n)}
\label{eq:bound_B_3}
\end{align}

Combing (\ref{eq:error_proba}), (\ref{eq:conditional_error}), (\ref{eq:error_1}), (\ref{eq:error_2}), (\ref{eq:error_3}), (\ref{eq:bound_B_1}), (\ref{eq:bound_B_2}) and (\ref{eq:bound_B_3}),  we can finally upper bound the average sum-decoding error probability over the ensemble as
\begin{align}
P_e^{(n)}&\leq  |\mathcal K_N|2^{-n(I(W;Y)-\epsilon_n)}+2^{-n(R_2+H(W|Y)-\log q+\delta_n)}+2^{-n(R_1+H(W|Y)-\log q+\delta_n)}+o(1)
\label{eq:error_final_bound}
\end{align}

To obtain a vanishing error probability, the second and third term in the above expression impose the constraints
\begin{subequations}
\begin{align}
R_1&< \log q-H(W|Y)=H(W)-H(W|Y)-H(W)+\log q=I(W;Y)-D(q)\\
R_2&< \log q-H(W|Y)=H(W)-H(W|Y)-H(W)+\log q=I(W;Y)-D(q)
\end{align}
\label{eq:bound_r1_r2}
\end{subequations}
Using the result of Theorem \ref{thm:size_sumset} on the size of $|\mathcal K_N|$, the following bounds are obtained.
\begin{align}
|\mathcal K_N|2^{-n(I(W;Y)-\epsilon_n)}&\leq\min\left\{2^{n(R_1+R_2)}, 2^{n\left(\max\{R_1,R_2\}+D(q)\right)}\right\}2^{-n(I(W;Y)-\epsilon_n)}\\
&=\min\left\{2^{-n(I(W;Y)-R_1-R_2-\epsilon_n)}, 2^{-n\left(I(W;Y)-\max\{R_1,R_2\}-D(q)-\epsilon_n\right)}\right\}
\end{align}
The above quantity can be made arbitrarily small if we have either
\begin{align}
R_1+R_2<I(W;Y)-\epsilon_n
\label{eq:bound_low}
\end{align}
or
\begin{align}
\max\{R_1,R_2\}<I(Y;W)-D(q)-\epsilon_n
\label{eq:bound_high}
\end{align}
To conclude, in order to make the sum-decoding error probability in (\ref{eq:error_final_bound}) arbitrarily small, we need the individual rate $R_1,R_2$ to satisfy the condition in (\ref{eq:bound_r1_r2}), and the sum rate $R_1+R_2$ to satisfy either $(\ref{eq:bound_low})$ or $(\ref{eq:bound_high})$. The intersection of all these constraints gives the claimed result.
\end{proof}

\appendices

\section{Some properties of $D(q)$}
\label{sec:Dq}

The fact that $D(q)$ is increasing with $q$ can be shown straightforwardly by checking $D(q+1)\geq D(q)$ for all $q\in\mathbb N^+$. The sum $\sum_{i=1}^qi\log i$ can be bounded as
\begin{align}
\int_{1}^q x\log x dx+1\cdot\log 1\leq \sum_{i=1}^qi\log i\leq \int_{1}^q x\log x dx+q\cdot\log q
\end{align}
which evaluates to
\begin{align*}
\frac{q^2}{2}\log q-\log e(q^2/4+1/4)\leq \sum_{i=1}^qi\log i\leq \frac{q^2}{2}\log q-\log e(q^2/4+1/4)+q\log q
\end{align*}
Using the expression in (\ref{eq:H_W}) we have
\begin{align*}
\log q+\log \sqrt{e}-\frac{1+q\log q}{q^2}\leq H(U_1+U_2)\leq \log q+\log\sqrt{e}-\frac{1-q\log q}{q^2}.
\end{align*}
This shows that for $q\rightarrow \infty$ we have $H(U_1+U_2)\rightarrow \log q+\log \sqrt{e} $ hence $D(q)\rightarrow \log \sqrt{e}$.

\section{Conditional expectation and variance of $Z_1$}
\label{app:Z_1}
Here we prove the claim used in the proof of Lemma \ref{lemma:size_p1} on the conditional expectation and variance of $Z_1$.

Recall that in the proof of Lemma \ref{lemma:size_p1} we defined $Z_1(\ve p)$ to be the number of message pairs $(\ve m,\ve n')$ such that ${\sf p}_1(\ve m,\ve n')=\ve p$ for some $\ve p\in \mathcal A$.  Furthermore, we will only consider the pairs $(\ve m,\ve n')$ such that $\sf s(\ve m,\ve n')=\ve s$ for some $\ve s$. In Lemma \ref{lemma:information_sum} we have shown that there are $L$ pairs of such $(\ve m,\ve n')$. We use ${\sf p}_{1}(\ell)$ to denote the parity sum ${\sf p}_1$ of the $\ell$-th pair $(\ve m,\ve n')$, for $\ell=1,\ldots, L$.

For the analysis in this section, we have the following local definitions. For a  given vector $\ve p\in\mathcal W^{(k_1-k_2)}$, define the random variables $Z_{\ell,i}(\ve p), i\in[1:n-k_1]$ to be the indicator function 
\begin{align}
Z_{\ell, i}(\ve p):=\ve 1 \{{\sf p}_{1,i}(\ell)=\ve p_{i}\}
\end{align}
i.e., $Z_{\ell, i}(\ve p)$ equals $1$ when the $i$-th entry of the parity-sum ${\sf p}_1(\ell)$ is equal to the entry $\ve p_{i}$. Furthermore we define
\begin{align}
Z_\ell(\ve p)&:=\prod_{i=1}^{k_1-k_2}Z_{\ell,i}(\ve p)
\end{align}
hence $Z_{\ell}(\ve p)$ is also an indicator function and is equal to $1$ if the $\ell$-th pair sums up to the parity-sum $\ve p$. Then we can define $Z_1(\ve p)$ as
\begin{align*}
Z_1(\ve p)&:=\sum_{\ell=1}^LZ_\ell(\ve p).
\end{align*}
which indeed counts the number of different pairs $(\ve m,\ve n')$ satisfying $\sf s(\ve m,\ve n')=\ve s$ and ${\sf p}_1(\ve m,\ve n)=\ve p$. With this notation the event 
$\{\sf p_1(\ell)=\ve p\}$
is equivalent to the event
$\left\{Z_{\ell}(\ve p)=1\right\}$
and the  following event
\begin{align}
\{\ve p\in\mathcal P_1(\ve s)\}=\{{\sf p}_1(\ell)=\ve p \mbox{ for some }\ell\in[1:L]\}
\end{align}
is equivalent to the event 
$\left\{Z_1(\ve p)\geq 1 \right\}$.  Notice that the dependence on the information-sum $\sf s$ is omitted in above notations.

We calculate the conditional expectation $\E{Z_1(\ve p)|
\sf s(\ve m,\ve n')=\ve s}$ and conditional variance $\var{Z_1(\ve p)|\sf s(\ve m,\ve n')=\ve s}$ for a typical sequence $\ve p\in \mathcal A$. Now for a sequence $\ve p\in\mathcal  A$, by definition we have
\begin{align}
\E{Z_1(\ve p)|\sf s(\ve m,\ve n')=\ve s}&=\sum_{\ell=1}^L\E{\prod_{i=1}^{k_1-k_2}Z_{\ell,i}(\ve p)\middle|\sf s(\ve m,\ve n')=\ve s}\\
&\stackrel{(a)}{=}\sum_{\ell=1}^L\prod_{i=1}^{k_1-k_2}\E{Z_{\ell,i}(\ve p)\middle|\sf s(\ve m,\ve n')=\ve s}\\
&=\sum_{\ell=1}^L\prod_{i=1}^{k_1-k_2}\pp{{\sf p}_{1,i}(\ell)=\ve p_i\middle|{\sf s}(\ve m,\ve n')=\ve s}
\end{align}
where step $(a)$ follows because $Z_{\ell,i}$ are also independent for different $i$. To see this, notice that $i$-th row of of the parity-sum $\sf p_1$ is of the form $\ve m_{2,i}\oplus \ve d_{2,i}+\ve H_i^T\ve n'$.  Since $\ve m_{2,i}$ is independent for each $i$ and each row $\ve H_i$ is chosen independently from the other rows,  $Z_{\ell,i}$ is also independent for different $i$.

Now we use the set $I(\ve p,a)$ to denote all indices of entries of $\ve p$ taking value the $a\in\mathcal W$.  For a given $\ve p\in\mathcal A$, we can rewrite the product term as:
\begin{align}
\prod_{i}^{k_1-k_2}\pp{ {\sf p}_{1,i}(\ell)=\ve p_i|\sf s(\ve m,\ve n')=\ve s }&=\prod_{a=0}^{2q-2}\prod_{i\in I(\ve p, a)}\pp{ {\sf p}_{1,i}(\ell)=a|\sf s(\ve m,\ve n')=\ve s }
\end{align}
Recall that $\sf s(\ve m,\ve n')=\ve m_1\oplus \ve d_1+\ve n'$. Since each entry of $\ve H$ is chosen i.i.d. uniformly at random from $\mathbf F_q$, and $\ve m_1$ are independent from $\ve m_2$, then $\sf s(\ve m,\ve n')$ is also independent from ${\sf p}_{1}(\ve m,\ve n')$. Hence we have
\begin{align*}
\pp{ {\sf p}_{1,i}(\ell)=a|\sf s(\ve m,\ve n')=\ve s }&=\pp{ {\sf p}_{1,i}(\ell)=a}\\
&=P_W(a)
\end{align*}
The last step follows from the fact that ${\sf p}_{1,i}(\ell)$ has distribution $P_W$ (established in Lemma \ref{lemma:check_distr}). We are concerned with the case when $\ve p$ is a typical sequence in $\mathcal A^{(k_1-k_2)}_{[W]}$ hence $|I(\ve p,a)|=(k_1-k_2)(P_W(a)+o(1))$. We can continue as
\begin{align}
E(Z_{\ell}(\ve p)|\sf s(\ve m,\ve n')=\ve s)&=\prod_{i=1}^{k_1-k_2}\pp{ {\sf p}_{1,i}(\ell)=\ve p_i|\sf s(\ve m,\ve n')=\ve s }\\
&= \prod_{a=0}^{2q-2}\pp{{\sf p}_{1,i}(\ell)=a|\sf s(\ve m,\ve n')=\ve s}^{|I(\ve p, a)|}\\
&= \prod_{a=0}^{2q-2}(P_W(a))^{|I(\ve p, a)|}\\
&=\prod_{a=0}^{2q-2}P_W(a)^{(k_1-k_2)(P_W(a)+o(1))}\\
&=2^{-(k_1-k_2)(H(W)+o(1))}
\end{align}
Notice that $\E{Z_{\ell}(\ve p)|\sf s(\ve m,\ve n')=\ve s}$ does not depend on $\ell$ asymptotically. Using Lemma \ref{lemma:information_sum} we have:
\begin{align}
\E{Z_1(\ve p)|\sf s(\ve m,\ve n')=\ve s}&=\sum_{\ell=1}^L\E{Z_\ell(\ve p)|\sf s(\ve m,\ve n')=\ve s}\\
&=L2^{-(k_1-k_2)(H(W)+o(1))}\\
&= 2^{R_1+R_2-R_2H(W)/\log q -(R_1-R_2)H(W)/\log q+o(1)}\\
&=2^{n(R_2-R_1D(q)/\log q+o(1))}
\end{align}

To evaluate the variance, we first observe  that (here we drop $\ve p$ for simplicity)
\begin{align}
Z_1^2&=\left(\sum_{\ell=1}^LZ_{\ell}\right)^2\\
&=\sum_{\ell=1}^L Z_{\ell}^2+\sum_{\ell\neq j}Z_{\ell}Z_{j}\\
&=\sum_{\ell=1}^L  Z_{\ell}+\sum_{\ell\neq j}Z_{\ell}Z_{j} \\
&=Z_1+\sum_{\ell\neq j}Z_{\ell}Z_{j}
\end{align}
as $Z_\ell^2=\prod_i Z_{\ell,i}^2=\prod_i Z_{\ell,i}=Z_{\ell}$ for indicator functions. Furthermore, using the fact that ${\sf p}_1(\ve m,\ve n')$ and $\sf s(\ve m,\ve n')$ are independent, we have
\begin{align}
\E{Z_1^2 | \sf s(\ve m,\ve n')=\ve s}&=\E{Z_1^2}\\
&= \E{ Z_1}+\sum_{\ell\neq j}\E{Z_{\ell}Z_{j} }\\
&\stackrel{(a)}{=}\E{Z_1}+\sum_{\ell\neq j}\E{Z_{\ell}}\E{Z_{j}}\\
&\leq \E{Z_1}+\E{Z_1}^2\\
&=\E{Z_1|\sf s(\ve m,\ve n')=\ve s}+\E{Z_1|\sf s(\ve m,\ve n')=\ve s}^2
\end{align}
where step $(a)$ follows because $Z_\ell, Z_j$ are  independent for $\ell\neq j$. Hence we have
\begin{align}
\var{Z_1| \sf s(\ve m,\ve n')=\ve s}&=\E{(Z_1-\E{Z_1|\sf s(\ve m,\ve n')=\ve s})^2|\sf s(\ve m,\ve n')=\ve s}\\
&=\E{Z_1^2|\sf s(\ve m,\ve n')=\ve s}-\E{Z_1|\sf s(\ve m,\ve n')=\ve s}^2\\
&\leq \E{Z_1|\sf s(\ve m,\ve n')=\ve s}+\E{Z_1|\sf s(\ve m,\ve n')=\ve s}^2-\E{Z_1|\sf s(\ve m,\ve n')=\ve s}^2\\
&= \E{Z_1|\sf s(\ve m,\ve n')=\ve s}
\end{align}

\section{Conditional expectation and variance of $Z_2$}
\label{app:Z_2}
Here we prove the claim used in the proof of Lemma \ref{lemma:size_p2} on the conditional expectation and variance of $Z_2$. The proof is similar to that in Appendix \ref{app:Z_1}.

Recall that in the proof of Lemma \ref{lemma:size_p2} we defined $Z_2(\ve p)$ to be the number of message pairs $(\ve m,\ve n')$ such that ${\sf p}_2(\ve m,\ve n')=\ve p$ for some $\ve p\in \mathcal A^{(n-k_1)}_{[W]}$.  Furthermore, we are only concerned with the pairs $(\ve m,\ve n')$ such that $\sf s(\ve m,\ve n')=\ve s$ and ${\sf p}_1(\ve m,\ve n')=\ve p_1$ for some $\ve s$ and $\ve p_1$. We also showed in (\ref{eq:L_p_approx}) that there are $L'$ pairs of such $(\ve m,\ve n')$. We use ${\sf p}_{2}(\ell)$ to denote the parity sum ${\sf p}_2$ of the $\ell$-th pair $(\ve m,\ve n')$, for $\ell=1,\ldots, L'$.

For the analysis in this section, we have the following local definitions, which are similar to the definitions in Appendix \ref{app:Z_1}.  For a given vector $\ve p\in\mathcal W^{(n-k_1)}$, we  define random variables $Z_{\ell,i}(\ve h), i\in[1:n-k_1]$ to be the indicator function 
\begin{align}
Z_{\ell, i}(\ve p):=\ve 1\{{\sf p}_{2,i}(\ell)=\ve p_i\}
\end{align}
i.e., $Z_{\ell, i}(\ve p)$ equals $1$ when the $i$-th entry of the parity-sum ${\sf  p}_2(\ell)$ is equal to the entry $\ve p_i$. Furthermore we define
\begin{align}
Z_\ell(\ve p)&:=\prod_{i=1}^{n-k_1}Z_{\ell,i}(\ve p)
\end{align}
hence $Z_{\ell}(\ve p)$ is also an indicator function and is equal to $1$ if the $\ell$-th pair sums up to the parity-sum $\ve p$. Then we can define $Z_2(\ve p)$ as
\begin{align}
Z_2(\ve p):=\sum_{\ell=1}^{L'}Z_\ell(\ve p).
\end{align}
which indeed counts the number of different pairs $(\ve m,\ve n')$ satisfying $\sf s(\ve m,\ve n')=\ve s$, ${\sf p}_1(\ve m,\ve n)=\ve p_1$ and ${\sf p}_2(\ve m,\ve n')=\ve p$. With this notation the event 
$\{{\sf p}_2(\ell)=\ve p\}$
is equivalent to the event
$\left\{Z_{\ell}(\ve p)=1\right\}$
and the  following event
\begin{align}
\{\ve p\in\mathcal P_2(\ve s,\ve p_1)\}=\{{\sf p}_2(\ell)=\ve p \mbox{ for some }\ell\in[1:L']\}
\end{align}
is equivalent to the event 
$\left\{Z_2(\ve p)\geq 1 \right\}$.  Notice that the dependence on the sum $\ve s$ and $\ve p_1$ is omitted in above notations.

We calculate the conditional expectation 
$\E{Z_2(\ve p)|
F(\ve a)}$ and conditional variance $\var{Z_2(\ve p)|F(\ve a)}$ for typical sequence $\ve p\in \mathcal F(\ve a)$. Notice we have $\ve p_i\in\{\ve a_i,\ve a_i+q\}$ conditioned on the event $F(\ve a)$ for some $\ve a_i\in[0:q-1]$. Now for a sequence $\ve p\in\mathcal F(\ve a)$, by definition we have
\begin{align}
\E{Z_2(\ve p)|F(\ve a)}&=\sum_{\ell=1}^{L'}\E{\prod_{i=1}^{n-k_1}Z_{\ell,i}(\ve p)\middle|F(\ve a)}\\
&\stackrel{(a)}{=}\sum_{\ell=1}^{L'}\prod_{i=1}^{n-k_1}\E{Z_{\ell,i}(\ve p)\middle|F(\ve a)}\\
&=\sum_{\ell=1}^{L'}\prod_{i=1}^{n-k_1}\pp{ {\sf p}_{2,i}(\ell)=\ve p_i\middle|F(\ve a)}
\end{align}
where step $(a)$ follows since each row  $\ve Q_i$ is picked independently, hence $Z_{\ell,i}$ are also independent for different $i$.

We again use the set $I(\ve p,a)$ to denote all indices of entries of $\ve p$ taking the value $a\in\mathcal W$.  For a given $\ve p$, we can rewrite the product term as:
\begin{align}
\prod_{i}^{n-k_1}\pp{{\sf p}_{2,i}(\ell)=\ve p_i|F(\ve a) }&=\prod_{a=0}^{2q-2}\prod_{i\in I(\ve p, a)}\pp{ {\sf p}_{2,i}(\ell)=a|F(\ve a) }
\end{align}
For any $i\in I(\ve p, a)$ and any $\ell\in[1:L']$, we have
\begin{align*}
\pp{ {\sf p}_{2,i}(\ell)=a|F(\ve a)}&=\frac{\pp{ {\sf p}_{2,i}(\ell)=a,F(\ve a)}}{\pp{F(\ve a)}}\\
&=\frac{\pp{{\sf p}_{2,i}(\ell)=a, {\sf p}_{2,j}(\ell)\in\{\ve a_j,\ve a_j+q\} \mbox{ for all }j\in [1:n-k_1]}}{\pp{{\sf p}_{2,j}(\ell)\in\{\ve a_j,\ve a_j+q\} \mbox{ for all }j\in[1:n-k_1]}}\\
&\stackrel{(a)}{=}\frac{\pp{{\sf p}_{2,i}(\ell)=a, {\sf p}_{2,i}(\ell)\in\{\ve a_i=a,\ve a_i+q=a+q\}}}{\pp{ {\sf p}_{2,i}(\ell)\in\{\ve a_i,\ve a_i+q\}}}\cdot \frac{\pp{ {\sf p}_{2,j}(\ell)\in\{\ve a_j,\ve a_j+q\} \mbox{ for all }j\neq i}}{\pp{ {\sf p}_{2,j}(\ell)\in\{\ve a_j,\ve a_j+q\} \mbox{ for all }j\neq i}}\\
&=\frac{\pp{ {\sf p}_{2,i}(\ell)=a}}{\pp{ {\sf p}_{2,i}(\ell)\in\{a,a+q\}}}\\
&=P_W(a)\cdot q
\end{align*}
where step $(a)$ follows from the fact that $\ve p\in\mathcal F(\ve a)$ and  $Z_{\ell,i}$ are  independent for different $i$. The last step follows from the fact that ${\sf p}_{2,i}(\ell)$ has distribution $P_W$ (established in Lemma \ref{lemma:check_distr}) and it is easy to see that $\pp{{\sf p}_{2,i}(\ell)\in\{a,a+q\}}=1/q$ for all $a\in[0:q-1]$.

We are concern with the case when $\ve p$ is a typical sequence in $\mathcal A^{(n-k_1)}_{[W]}$ hence $|I(\ve p, a)|=(n-k)(P_W(a)+o(1))$. We can continue as
\begin{align}
E(Z_{\ell}(\ve p)|F(\ve a))&=\prod_{i=1}^{n-k_1}\pp{ {\sf p}_{2,i}(\ell)=\ve p_i|F(\ve a) }\\
&= \prod_{a=0}^{2q-2}\pp{{\sf p}_{2,i}(\ell)=a|F(\ve a)}^{|I(\ve p, a)|}\\
&= \prod_{a=0}^{2q-2}(P_W(a)\cdot q)^{|I(\ve p, a)|}\\
&=q^{\sum_{a=0}^{2q-2}|I(\ve p,a)|}\prod_{a=0}^{2q-2}P_W(a)^{(n-k+1)(P_W(a)+o(1))}\\
&=q^{n-k_1}\cdot 2^{(n-k_1)(-H(W)+o(1))}\\
&=2^{(n-k_1)(\log q-H(W)+o(1))}
\end{align}
Notice that $\E{Z_{\ell}(\ve p)|F(\ve a)}$ does not depend on $\ell$ asymptotically. Using $L'$ given in (\ref{eq:L_p_approx})  we have:
\begin{align}
\E{Z_2(\ve p)|F(\ve a) }&=\sum_{\ell=1}^{L'}\E{Z_\ell(\ve p)|F(\ve a)}\\
&=L'2^{(n-k_1)(\log q-H(W)+o(1))}\\
&=2^{n(R_2-R_1D(q)/\log q) }2^{n(R_1D(q)/\log q-D(q)+o(1))}\\
&=2^{n(R_2-D(q)+o(1))}
\end{align}

To evaluate the variance, by the same argument in the proof in Appendix \ref{app:Z_1} we have
\begin{align}
\E{Z_2^2\middle |F(\ve a)}&= \E{ Z_2 \middle |F(\ve a)}+\sum_{\ell\neq j}\E{Z_{\ell}Z_{j}\middle | F(\ve a)}\\
&\stackrel{(a)}{=}\E{Z_2|F(\ve a)}+\sum_{\ell\neq j}\E{Z_{\ell}\middle | F(\ve a)}\E{Z_{j}\middle | F(\ve a)}\\
&\leq \E{Z_2|F(\ve a)}+\E{Z_2|F(\ve a)}^2
\end{align}
where step $(a)$ follows because $Z_\ell, Z_j$ are conditionally independent for $\ell\neq j$, conditioned on the event $F(\ve a)$.
Hence we have
\begin{align}
\E{(Z_2-\E{Z_2|F(\ve a)})^2|F(\ve a)}&=\E{Z_2^2|F(\ve a)}-\E{Z_2|F(\ve a)}^2\\
&<\E{Z_2|F(\ve a)}+\E{Z_2|F(\ve a)}^2-\E{Z_2|F(\ve a)}^2\\
&= \E{Z_2|F(\ve a)}
\end{align}

\section{On the distribution of parity-sums}
\label{appendix:parity_sum}

Given randomly chosen  message pairs $(\ve m, \ve n')$, we analyze the distribution of the parity-sum ${\sf p}_1$ and ${\sf p}_2$ when the matrices $\ve Q, \ve H$ are chosen randomly.

\begin{lemma}[Distribution of parity-sum]
Let $(\ve m,\ve n')$ be two messages which are independently and uniformly chosen at random from $\mathbb F_q^{k_1}$ and $\mathbb F_q^{k_2}$ respectively. As in (\ref{eq:sum_systematic}), define the parity-sum ${\sf p}_1(\ve m,\ve n'):=\ve m_2\oplus \ve d_2+\ve H\ve n'$ and ${\sf p}_2(\ve m,\ve n'):= \ve Q\ve m\oplus\ve d_3+\ve Q\ve n$ with $\ve n:=\begin{bmatrix}
\ve n'\\
\ve H\ve n'
\end{bmatrix}$ for any fixed $\ve d_2,\ve d_3$. We assume that each entry of $\ve Q, \ve H$ are chosen independently and uniformly from $\mathbb F_q$. Then each entry of ${\sf p}_1$ and ${\sf p}_2$ is independent and has the distribution $p_W$ defined in (\ref{eq:P_W}).
\label{lemma:check_distr}
\end{lemma}

\begin{proof}
We first consider the parity-sum ${\sf p_1}$. Since $\ve m$ are chosen uniformly from $\mathbb F_q^{k_1}$ and each row $\ve H$ is independently chosen, each entry of ${\sf p_1}$ is also independent.   The $i$-th entry of $\ve H\ve n'$ is of the form
\begin{align*}
\ve H_{i1}\ve n_1'\oplus \ve H_{i2}\ve n'_2\oplus \ldots\oplus \ve H_{ik_2}\ve n'_{k_2}
\end{align*}
which is uniformly distributed in $\mathbb F_q$ for large $k_2$ since each $H_{ij}$ is chosen independently uniformly from $\mathbb F_q$. Furthermore since each entry of $\ve m_2$ is i.i.d. in $\mathbb F_q$, then each entry of ${\sf p}_1$ is i.i.d. according to $p_W$.

For each entry of the parity-sum ${\sf p_2}$, we write out its $i$-th entry explicitly as $\ve Q_i^T\ve m\oplus \ve d_{3,i}+\ve Q_i^T\ve n$ where
\begin{align}
\ve Q_i^T\ve m\oplus \ve d_{3,i}&=\ve Q_{i1}\ve m_{1}\oplus \cdots\oplus \ve Q_{ik_1}\ve m_{k_1}\oplus \ve d_{3,i}\\
\ve Q_i^T\ve n&=\ve Q_{i1}\ve n_{1}\oplus \cdots\oplus \ve Q_{ik_1}\ve n_{k_1}
\end{align}
Since each row $\ve Q_i$ is independently chosen and both $\ve Q_i^T\ve m$ and $\ve Q_i^T\ve n$ has the uniform distribution in $\mathbb F_q$ for large $k_1$, we also conclude that each entry of ${\sf p}_2$ is i.i.d. according to $p_W$.
\end{proof}

\section{Derivations in the proof of Theorem \ref{thm:computation_rate}}
\label{sec:independent}

Here we prove the statement  in the proof of Theorem \ref{thm:computation_rate}. Recall that $M=\ve m_1$, $N=\ve n_1$ are two different chosen messages and $\tilde W^n(\ve m',\ve n'):=G\ve m'\oplus d_1+G\ve n'\oplus d_2\neq G\ve m_1\oplus d_1+G\ve n_1\oplus d_2$ where $G$ and $d_1, d_2$ are the randomly chosen generator matrix and dither vectors.  To lighten the notation, in this section we define $U_1(\ve m):=G\ve m\oplus d_1$ and $U_2(\ve n):=G\ve n\oplus d_2$.

We first prove that for $\tilde W(\ve m',\ve n')\in\mathcal B_1$, we have
\begin{align}
\pp{\tilde W^n(\ve m',\ve n')=\ve w_0|Y^n=\ve y,\mathcal M}=\pp{\tilde W^n(\ve m',\ve n')=\ve w_0\mid \mathcal M}.
\end{align}
which is equivalent to
\begin{align}
\pp{\tilde W^n(\ve m',\ve n')=\ve w_0,Y^n=\ve y| \mathcal M}=\pp{\tilde W^n(\ve m',\ve n')=\ve w_0|\mathcal M}\pp{Y^n=\ve y|\mathcal M}
\label{eq:inde_equivalent}
\end{align}

This is shown straightforwardly as
\begin{align}
&\pp{\tilde W^n(\ve m',\ve n')=\ve w_0,Y^n=\ve y|\mathcal M}\\
&=\sum_{\ve t,\ve v}\pp{\tilde W^n(\ve m',\ve n')=\ve w_0,Y^n=\ve y, U_1(\ve m_1)=\ve t,U_2(\ve n_1)=\ve v|\mathcal M}\\
&=\sum_{\ve t,\ve v}\pp{\tilde W^n(\ve m',\ve n')=\ve w_0, U_1(\ve m_1)=\ve t,U_2(\ve n_1)=\ve v|\mathcal M}\pp{Y^n=\ve y|\tilde W^n(\ve m',\ve n')=\ve w_0, U_1(\ve m_1)=\ve t,U_2(\ve n_1)=\ve v,\mathcal M}\\
&=\sum_{\ve t,\ve v}\pp{W'(\ve m',\ve n')=\ve w_0, U_1(\ve m_1)=\ve t,U_2(\ve n_1)=\ve v|\mathcal M}\pp{Y^n=\ve y|U_1(\ve m_1)=\ve t,U_2(\ve n_1)=\ve v,\mathcal M}\\
&=\sum_{\ve t,\ve v}\sum_{\stackrel{\ve t',\ve v'}{\ve t'+\ve v'=\ve w_0}}\pp{U_1(\ve m')=\ve t',U_2(\ve n')=\ve v',U_1(\ve m_1)=\ve t,U_2(\ve n_1)=\ve v|\mathcal M}\pp{Y^n=\ve y|U_1(\ve m_1)=\ve t,U_2(\ve n_1)=\ve v|\mathcal M}
\label{eq:inde_substitute}
\end{align}
In this case we have
\begin{align}
&\sum_{\stackrel{\ve t',\ve v'}{\ve t'+\ve v'=\ve w_0}}\pp{G\ve m'\oplus d_1=\ve t',G\ve n'\oplus d_2=\ve v',G\ve m_1\oplus d_1=\ve t,G\ve n_1\oplus d_2=\ve v|\mathcal M}\\
&\stackrel{(a)}{=}\sum_{\stackrel{\ve t',\ve v'}{\ve t'+\ve v'=\ve w_0}}\pp{G\ve m'\oplus d_1=\ve t',G\ve n'\oplus d_2=\ve v'}\pp{G\ve m_1\oplus d_1=\ve t,G\ve n_1\oplus d_2=\ve v}\\
&=\pp{U_1(\ve m')+U_2(\ve n')=\ve w_0|\mathcal M}\pp{U_1(\ve m_1)=\ve t,U_2(\ve n_1)=\ve v|\mathcal M}
\label{eq:independent}
\end{align}
where $(a)$ holds because for randomly chosen $G, d_1, d_2$ and the assumption that  $\ve m', \ve n'$ are different from $\ve m_1,\ve n_1$ and linearly independent, the random variables $(U_1(\ve m'), U_2(\ve n'))$ are independent  from $(U_1(\ve m_1), U_2(\ve n_1))$. 
Substituting it back to (\ref{eq:inde_substitute}) we have
\begin{align*}
&\pp{W'(\ve m',\ve n')=\ve w_0,Y^n=\ve y|\mathcal M}\\
&=\sum_{\ve t,\ve v}\pp{W'(\ve m',\ve n')=\ve w_0|\mathcal M}\pp{U_1(\ve m_1)=\ve t,U_2(\ve n_1)=\ve v|\mathcal M}\pp{Y^n=\ve y|U_1(\ve m_1)=\ve t,U_2(\ve n_1)=\ve v,\mathcal M}\\
&=\pp{W'(\ve m',\ve n')=\ve w_0|\mathcal M}\pp{Y^n=\ve y|\mathcal M}
\end{align*}
which proves the claim in (\ref{eq:inde_equivalent}).

We prove that for $\tilde W^{n}(\ve m_1, \ve n')\in\mathcal B_2$, we have
\begin{align}
\pp{\tilde W^n(\ve m_1,\ve n')=\ve w_0|Y^n=\ve y,\mathcal M}=q^{-n}
\end{align}
Using the same derivation as above, we arrive at
\begin{align}
&\pp{\tilde W^n(\ve m_1,\ve n')=\ve w_0,Y^n=\ve y|\mathcal M}\\
&=\sum_{\ve t,\ve v}\sum_{\stackrel{\ve v'}{\ve t+\ve v'=\ve w_0}}\pp{U_2(\ve n')=\ve v',U_1(\ve m_1)=\ve t,U_2(\ve n_1)=\ve v|\mathcal M}\pp{Y^n=\ve y|U_1(\ve m_1)=\ve t,U_2(\ve n_1)=\ve v,\mathcal M}
\end{align}
Furthermore we have
\begin{align}
&\pp{G\ve n'\oplus d_2=\ve v',G\ve m_1\oplus d_1=\ve t,G\ve n_1\oplus d_2=\ve v|\mathcal M}\\
&=\pp{G\ve m_1\oplus d_1=\ve t,G\ve n_1\oplus d_2=\ve v}\pp{G\ve n'\oplus d_2=\ve v'|G\ve m_1\oplus d_1=\ve t,G\ve n_1\oplus d_2=\ve v}\\
&=\pp{G\ve m_1\oplus d_1=\ve t,G\ve n_1\oplus d_2=\ve v}\pp{G\ve n'\oplus d_2=\ve v'}
\end{align}
The last equality holds because for $\ve n'$ different from $\ve n_1$ and $\ve m_1$ (we assume $\tilde W^n\notin\mathcal L$ hence $\ve n'$ cannot be equal to $\ve m_1$), $U_2(\ve n')$ is independent from $U_1(\ve m_1), U_2(\ve n_1)$ if $G, d_1, d_2$ are chosen randomly. Using the fact that $\pp{G\ve n'\oplus d_2=\ve v'}=q^{-n}$, we have
\begin{align}
\pp{\tilde W^n(\ve m_1,\ve n')=\ve w_0,Y^n=\ve y|\mathcal M}&=\sum_{\ve t,\ve v}\sum_{\stackrel{\ve v'}{\ve t+\ve v'=\ve w_0}}q^{-n} \pp{Y^n=\ve y, U_1(\ve m_1)=\ve t, U_2(\ve n_1)=\ve v|\mathcal M}\\
&=q^{-n}\pp{Y^n=\ve y|\mathcal M}
\end{align}
Finally we conclude that
\begin{align}
\pp{\tilde W^n(\ve m_1,\ve n')=\ve w_0|Y^n=\ve y,\mathcal M}=q^{-n}
\end{align}

\section*{Acknowledgment}

The authors wish to thank Sung Hoon Lim for many helpful discussions.


\bibliographystyle{IEEEtran}
\bibliography{IEEEabrv,SumsetLinearCodes}

\end{document}